\documentclass[letterpaper, 10 pt, conference]{ieeeconf}  

\pdfoutput=1

\usepackage{amsfonts, xfrac}
\usepackage{subfigure}
\usepackage{url}

\newcommand{\E}{\mathrm E}
\newcommand{\G}{\mathcal G}
\renewcommand{\P}{\mathrm P}
\newcommand{\Path}{\mathcal P}
\newcommand{\Ps}{\mathcal P_{s}}

\newcommand{\R}{\mathbb R}

\newcommand{\eps}{\varepsilon}

\newtheorem{definition}{Definition}
\newtheorem{lemma}{Lemma}
\newtheorem{theorem}{Theorem}

\newtheorem{proposition}{Proposition}

\newcommand{\var}{\mbox{var}}
\newcommand{\argmin}{\mbox{argmin}}
\newcommand{\argmax}{\mbox{argmax}}

\newcommand{\defeq}{\stackrel{{\rm def}}{=}}

\renewcommand{\hat}{\widehat}
\renewcommand{\tilde}{\widetilde}

\usepackage{mathtools}

\usepackage{algorithm}
\usepackage[noend]{algpseudocode}

\makeatletter
\def\BState{\State\hskip-\ALG@thistlm}
\makeatother

\begin{document}
\title{\Large \bf Observer Placement for Source Localization: \\ The Effect of Budgets and
Transmission Variance} 

\author{Brunella Spinelli, L.~Elisa Celis, Patrick Thiran}

%
%
%

\maketitle

\begin{abstract}

When an epidemic spreads in a network, a key question is where was its \emph{source}, i.e., the node that started the epidemic. 
If we know the time at which various nodes were infected, we can attempt to use this information in order to identify the source. 
However, maintaining \emph{observer} nodes that can provide their infection time may be costly, and we may have a \emph{budget} $k$ on the number of observer nodes we can maintain. 
Moreover, some nodes are more informative than others due to their location in the network. 
Hence, a pertinent question arises: \emph{Which nodes should we select as observers in order to maximize the probability that we can accurately identify the source?}

Inspired by the simple setting in which the node-to-node delays in the transmission of the epidemic are deterministic, we develop a principled approach for addressing the problem even when transmission delays are random.
We show that the optimal observer-placement differs depending on the \emph{variance} of the transmission delays and propose approaches in both low- and high-variance settings.
We validate our methods by comparing them against state-of-the-art observer-placements and show that, in both settings, our approach identifies the source with higher accuracy.
\end{abstract}



\section{Introduction}

Regardless of whether a network comprises computers, individuals or
cities, in many applications we want to detect whenever any anomalous or
malicious activity spreads across the network and, in particular, where the activity
originated.\footnote{In effect, we wish to answer questions such as \emph{what was the origin of a worm in a computer network?}, \emph{who was the instigator of a false rumor in a social
network?} and \emph{can we identify patient zero of a virulent disease?}  }
We call the spread of such activity an \emph{epidemic} and the originator the \emph{source}.

Clearly, monitoring all nodes is not feasible due to cost and overhead constraints: The number of nodes in the network may be prohibitively large
and some of them may be unable or unwilling to provide information about their state. 
Thus, studies have focused on how to estimate the source based on information from a few nodes (called \emph{observers}). 
Given a set of observers, many models and estimators for source localization have been developed \cite{Pinto12, Louni14, Zhang2016}. 
%
However, the \emph{selection} of observers has not yet received a satisfactory answer: 
Most of state-of-the-art methods are based on common centrality heuristics (e.g., degree- or betweenness-centrality) or on more advanced heuristic approaches that do not directly optimize source localization (see \cite{Zhang2016} for a survey) or are limited to simple networks such as trees (e.g., \cite{celis15}). 
%
Moreover, such methods consider only the structure of the network when placing observers.
However, depending on the particular epidemic, 
the expected transmission delay between two nodes, and its variance, can differ widely. 
%
We show that different transmission models require different observer placements:
This is illustrated in Figure \ref{fig:scheme}: As the variance of the transmission delays changes, the optimal set of observers also changes (see also Figure \ref{fig:ex_transition} for a concrete example).

\begin{figure}[H]
\begin{center}
{\includegraphics[width=0.6\columnwidth]{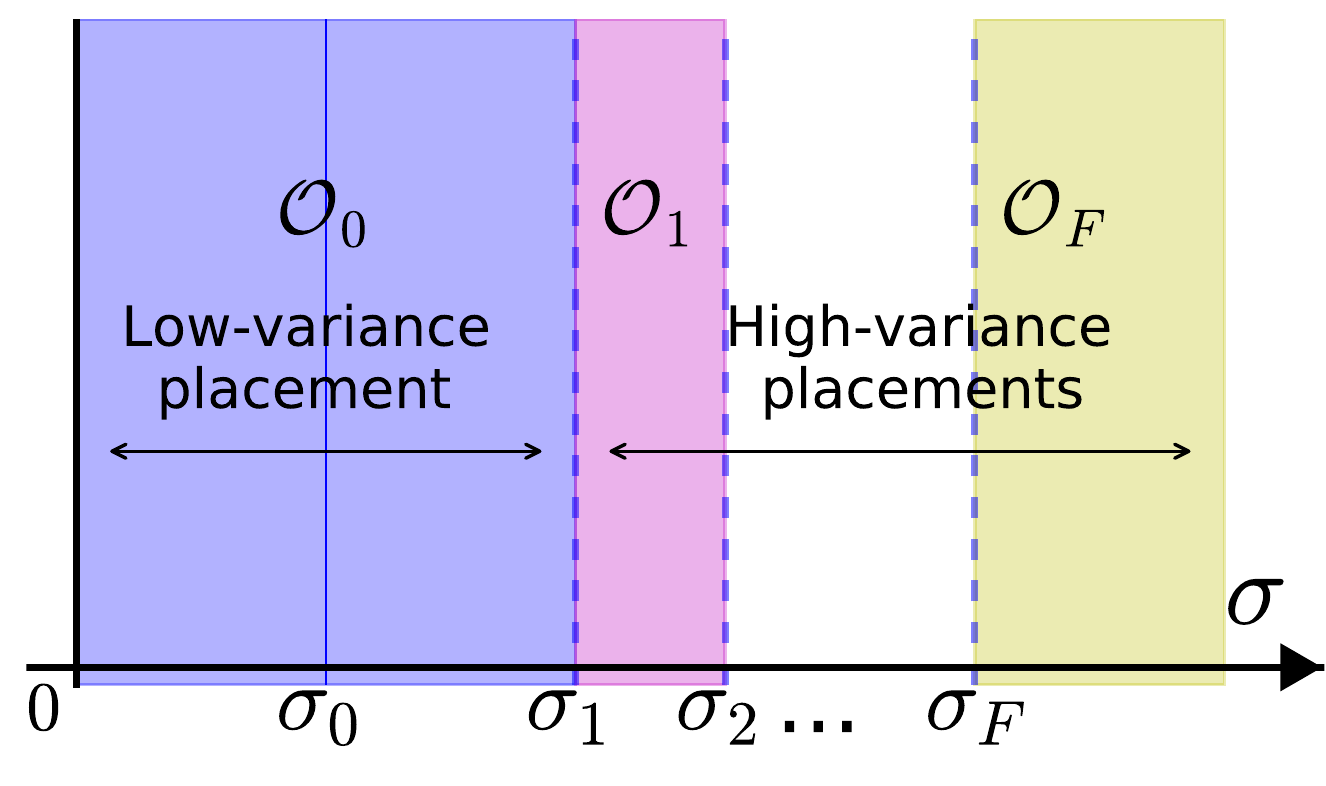}}
\caption{Transmission variance $\sigma$ and optimal observer placement. For ${\sigma \in (0, \sigma_0)}$ the transmission delays are effectively deterministic. For ${\sigma \in (\sigma_0, \sigma_1)}$ the variance $\sigma$ affects the accuracy of source localization but the optimal observer placement is still $\mathcal{O}_0$. For larger $\sigma$, the optimal observer placement may change, possibly multiple times ($\mathcal{O}_k$ denotes the optimal placement for ${\sigma \in (\sigma_k, \sigma_{k+1})}$) up to $\sigma=\sigma_F$. For $\sigma > \sigma_F$ the optimal placement remains the same ($\mathcal{O}_F$).}\label{fig:scheme}
\end{center}
\end{figure}

The difficulties faced in finding the optimal observers are two-fold. First, computing the likelihood of a node being the source conditional on the available observations can be computationally prohibitive \cite{Shah, Pinto12}; evaluating the probability of detection given a set of observers is, in general, even harder. Second, the optimal selection of a limited number of observers is NP-hard, even when the transmission times are deterministic. 
We take a principled approach that begins with considering deterministic transition delays, and build on this intuition in order to develop heuristics for both the low-variance and high-variance regimes.



\subsection{\bf Model and Problem Statement}

\textbf{Our Transmission Model.} We assume that the contact network $\G = (V,E, w)$ is known and is \emph{weighted}.
The weight $w_{uv} \in \mathbb R_{+}$ of edge $uv \in E$ is the mean of the
\emph{transmission delay} encoded by the random variable $X_{uv}$; this is the
time it would take for $u$ to infect $v$.\footnote{For ease of presentation we
assume the graph is undirected and $w_{uv}=w_{vu}$; however our definitions and approach
extend straightforwardly to the directed case.}
This transmission model is both natural and versatile as it comprises deterministic
transmissions (i.e., if $X_{uv}=w_{uv} \in \mathbb{R}_{+}$ a.s. for all edges $uv \in E$), which we call \emph{zero-variance}, and arbitrary \emph{random} independent transmission models. It naturally captures the SI epidemic
model adopted, e.g., in~\cite{Pinto12, luo2012} and related SIR/SIS/SEIR 
models (see~\cite{Krishnasamy2014} and the discussion in~\cite{Zhang2016}). 
We study, in particular, how the \emph{amount} of randomness (i.e., the variance of $X_{uv}$) in
the transmission delays affects the choice of observers for source
localization. Towards this, we
are the first to separately analyze two different regimes for the amount of 
randomness in transmission delays: \emph{low-variance} and
\emph{high-variance}. A dichotomy exists between the two, and our
approach for observer placement differs. 


\textbf{Our Source Estimation.} We assume that there is a single source that initiates the epidemic\footnote{
Our results can be extended to the case of multiple sources following the recent work by~\cite{Zhang-res} on a related problem. } 
and let $\mathcal{O} \subseteq V$ (which we will select) be the set of {observer} nodes. We assume we know the time at which each observer is infected, and refer to this vector of {infection times} as $T_{\mathcal{O}}$.
This is a standard (see, e.g.,~\cite{Netrapalli2011}) and realistic assumption (for example, clinics keep records of patients and carefully record outbreaks so can provide such information). To identify the source, we use this (and only this) information. 

We use maximum likelihood estimation (MLE) to produce an estimate 
$\hat s$ of the true unknown source $s^*$ as in ~\cite{Pinto12},\footnote{This approach is common 
(see e.g.,~\cite{Shah, dong2013}), although the exact form of the estimator
depends on the model and assumptions.}
 i.e., 
\begin{equation*}
\hat{s}\in \argmax_{\substack{s \in V}} \P(T_{\mathcal{O}}|s^*=s)\P(s^*=s).
\end{equation*}
We assume the prior on $s^*$ is uniform unless otherwise specified (i.e., $\P(s^*=s)=1/n$ for all nodes $s \in V$ where $n = |V|$). 



\textbf{Our Observer Placement.}
We assume that we are given a \emph{budget} $k$ on the number of observers we can
use, and that we must select our observers \emph{once and for all}. 
In order to select the \emph{best set of
observers $\mathcal{O}$ of size $k$} we must first define our metric of interest. We
consider the two metrics proposed by~\cite{celis15}, although variations
(including worst-case versions) exist~\cite{celis15}: 
\begin{enumerate} \setlength\itemsep{0em}
\item the \emph{success probability} $\Ps=\P(\hat{s} = {s^*})$, and 
\item the \emph{expected distance} between estimated source and real source, i.e.,
$\E[d(s^*, \hat{s})]$ with $d$ denoting the distance 
between two nodes in the network.
\end{enumerate} 
\noindent The two metrics might require different sets of observers~\cite{celis15},
however we show experimentally that maximizing $\Ps$ is a good proxy for
minimizing $\E[d(s^*, \hat{s})]$ (see Section~\ref{sec:low-variance}). Hence,
due to space constraints, we focus on the minimization of the former.

\subsection{\bf Main Contributions}

\textbf{Low-Variance Regime.} When the variance in the transmission delays is
\emph{low} (see Section~\ref{sec:low-variance}), 
we prove that the set of optimal observers is equal to the optimal set for the 
zero-variance regime. In the zero- and low- variance regime, both the probability of
success $\Ps$ and the expected distance $\E[d(s^*, \hat{s})]$ can be explicitly computed. 
Despite this seeming simplicity, the problem remains NP-hard. We tackle the
problem by using its connection with the well-studied
related Double Resolving Set (\emph{DRS}) problem~\cite{Caceres07} that minimizes the number of
observers for perfect detection.\footnote{This minimum number is, in many
cases, still prohibitively large, and can be as much as $n-1$, hence we
cannot use this approach directly.} 
From this connection we find inspiration for our algorithm that, by selecting one observer at a time until the budget is exhausted in order to reach a \emph{DRS} set, greedily
improves $\Ps$. 


\textbf{High-Variance Regime.} When the noise in the transmission delays is \emph{high} (see Section~\ref{sec:high-variance}), it is no longer negligible and it poses an additional challenge to source localization; in effect, the
accumulation of noise from
node to node as the epidemic spreads might no longer enable us to distinguish
between two potential sources, especially when they are both \emph{far} from all
observers. Hence, we must \emph{strengthen} the requirements for observer
placement in order to ensure that the nodes can be distinguished by observers that
are \emph{near} to them; this nearness is a function of the noise, the budget $k$, 
and the network topology. We define a novel objective function that both maximizes the
success probability and imposes a \emph{uniform} spread of observers in the
network. Taking inspiration from the low-variance regime, we design an algorithm
that greedily maximizes this new objective (see Section~\ref{sec:high-variance}). 

\begin{figure}[H]
\begin{center}
\subfigure[]{\includegraphics[width=0.48\columnwidth]{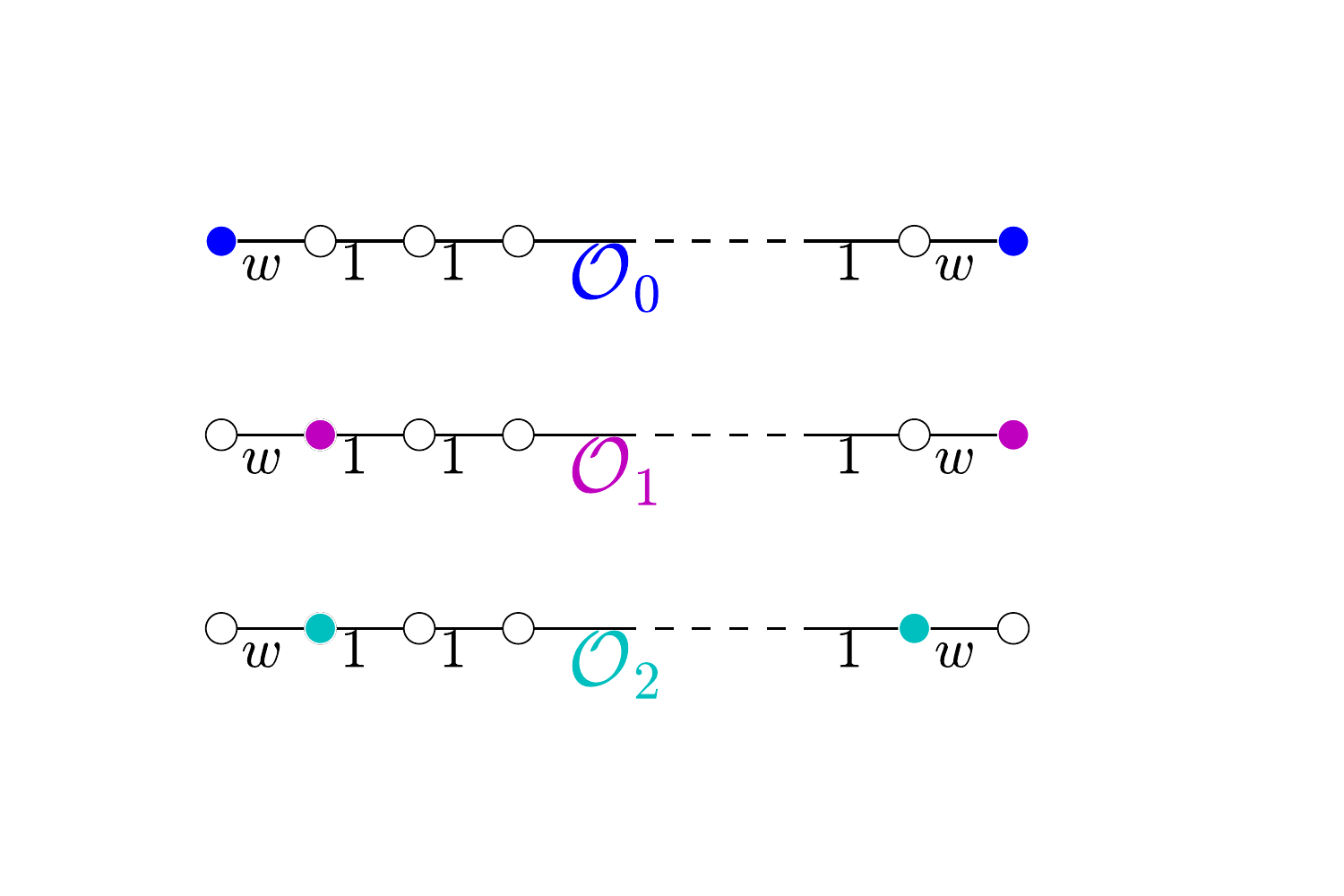}}
\subfigure[]{\includegraphics[width=0.48\columnwidth]{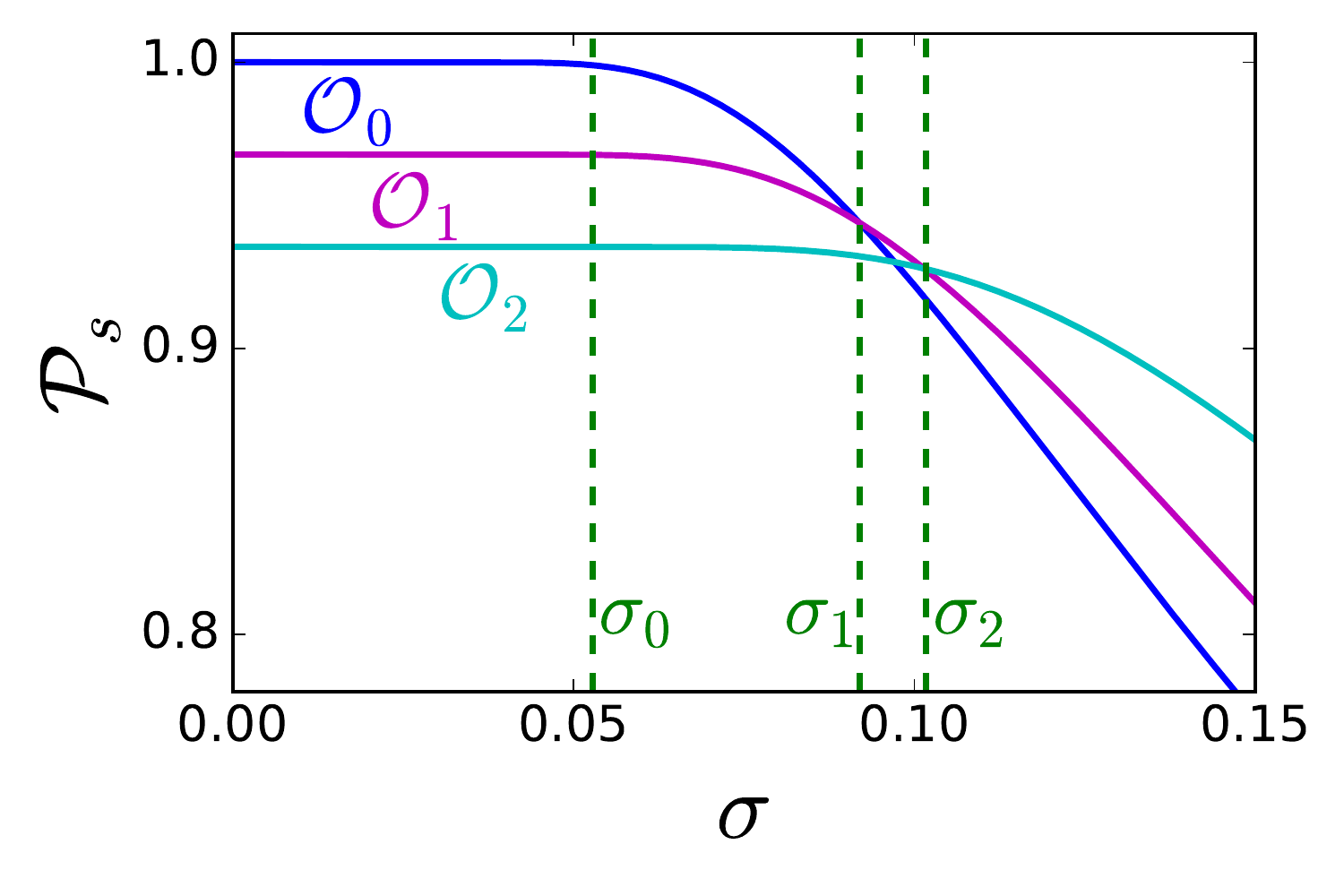}}
\caption{Optimal observers for Gaussian transmission delays with variance $\sigma^2$. (a): different observer placements; (b): their performance in terms of probability of success ($\mathcal{P}_s$) for $w=20$ and $30$ edges.}\label{fig:ex_transition}
\end{center}
\end{figure}

\textbf{Empirical Results.} 
Our methods perform favourably against state-of-the-art
approaches in both the low- and high-variance regimes (see
Section~\ref{sec:exp_bench}). In Appendix~\ref{app:alternate-obj}, for the low-variance regime, we further compare them against two other natural objective functions; we show that our approach is the best.
Moreover, in the empirical results the dichotomy between the low- and high-variance regimes becomes apparent.

 


\section{Related Work}\label{sec:related_work}

The problem of source localization has been widely studied in recent
years, we survey the works that are more relevant to ours and refer the reader to the survey by Jiang et. al.~\cite{jiang-survey} for a more complete review of the different approaches to source localization. 

\textbf{Transmission delays.} Many different transmission models for epidemics have been studied~\cite{Lelarge2009} and also considered in the context of source
localization. Although discrete-time transmission
delays are common~\cite{luo2013, prakash, Altarelli2014}, in order to better approximate realistic settings, much work (including ours) adopt continuous-time
transmission models with 
varying distributions for the transmission delays; e.g., exponential~\cite{Shah, luo2012} (primarily to capture percolation models such as
SI/SIS/SIR) or 
Gaussian~\cite{Pinto12,Louni14, Louni15, Zhang2016} (to capture delays that are
concentrated symmetrically around a value). In the same line of the latter
class of works, we use \emph{truncated} Gaussian variables, which gives us the advantage of ensuring that the model admits only strictly positive infection delays.

\textbf{Source localization.}
Many approaches~\cite{zheng2015, prakash, sundareisan2015hidden}, beginning with the seminal work by Shah and Zaman~\cite{Shah}, 
rely on knowing the state of the \emph{entire} network at a fixed point in
time $t$; this is often called a \emph{complete observation} of the epidemic.
These models use maximum likelihood estimation (MLE) to estimate the source. 
The results of~\cite{Shah} have been extended in many ways, for example in the case of
multiple sources~\cite{luo2012} or to obtain a \emph{local} source estimator~\cite{dong2013}. 
An alternate line of work considers a complete observation of the epidemic, except that the observed states are \emph{noisy}, i.e., potentially inaccurate~\cite{zhu2013, sundareisan2015hidden}.
As assuming the knowledge of the state of all the nodes is often not realistic, \emph{partial observation} settings have also been studied. In such a setting, only a subset of nodes $\mathcal O$
reveal their state. In this line of work, the observers are mainly \emph{given}, either arbitrarily or via a random process, and the problem of \emph{selecting} observers is not addressed. 
For example, when a fraction $x$ of nodes are randomly selected, Lokhov et al.~\cite{lokhov2014} propose an algorithm that relies on the knowledge of the state (S, I or R) of a fraction of the nodes in the graph at a given moment in time. This approach, however, crucially relies on the assumption that the starting time of the epidemic is known, which is often not realistic~\cite{jiang-survey, Pinto12}.
When the nodes are independently selected to be observers, 
an approach to source estimation based on the notion of \emph{Jordan center} was proposed~\cite{luo2013} and has since been used in other work for source estimation, especially with regard to a \emph{game theoretic} version of epidemics~\cite{Fanti2015}. 
This line of work 
does not assume infection times are known, which we believe is, in many cases, an unnecessary limitation. Indeed by using infection times we can achieve exact source localization in the zero-variance setting with sufficiently many observers~\cite{ChenHW14}, whereas this is not true otherwise. 

\textbf{Observer placement.}
Natural heuristics for observer placement (e.g., using high-degree vertices or optimizing for distance centrality) were first evaluated under the additional assumption that infected nodes know which neighbor infected them~\cite{Pinto12}.
Later, Luoni et al.~\cite{Louni14} proposed, for a similar model, to place the
observers using a Betweenness-Centrality criterion (which we use as a
benchmark, see Section~\ref{sec:exp_bench}), and extended it to noisy observations~\cite{Louni15}.
These and other heuristic approaches for observer placement are evaluated empirically by Seo et al.~\cite{Seo12}; they reach the conclusion that, among the placements they evaluate, the Betweenness-Centrality criterion performs the best. In their work the source is estimated by ranking candidates according to their distance to the set of observers, without using the time at which the observers became infected. Once again, this approach is inherently limited by the fact that it does not make use of the time of infection. 
The problem of \emph{minimizing} the number of observers required to detect the precise source (as opposed to \emph{maximizing} the performance given a \emph{budget} of observers) has been considered in the zero-variance setting. 
For trees, given the time at which the epidemic starts, the minimization problem was solved by Zejnilovic et al.~\cite{Zejn13}. 
Without assuming a tree topology and a known starting time, approximation algorithms have been developed towards this end~\cite{ChenHW14} (still in a zero-variance setting). However, in a network of size $n$, the number of observers required, even if minimized, can be up to $n-1$, hence, a budgeted setting is practically more interesting.
For trees, the budgeted placement of observers was solved by using techniques different from ours~\cite{celis15}. However these techniques heavily rely on the tree structure of the network and do not seem to be extendible to other topologies. 
In a very recent work, Zhang et al.~\cite{Zhang2016} consider selecting a fixed number of observers using several heuristics such as Betweenness-Centrality, Degree-Centrality and Closeness-Centrality and they show that none of these methods are satisfactory. 
They introduce a new heuristic for the choice of observers, called \emph{Coverage-Rate}, which is linked to the total number of nodes neighboring observers, and show that an approximated optimization of this metric yields better performance. Connecting the budgeted placement problem to the un-budgeted minimization problem, we provably outperform their approach in low-variance settings.\footnote{For example, on cycles of odd-length $d$ with a budget $k=2$ in the low-variance setting, any two nodes at distance more than $2$ are equivalent with
respect to the coverage rate, but only optimal if the observers are at distance
$(d-1)/2$; our approach selects this optimal placement.} 
Moreover, the effect of the variance in the transmission delays is neglected by Zhang et al., leaving open the question of whether their approach works in general. However, we consider Coverage-Rate as one of our baselines.

\textbf{Other related work.}
Finally, a closely related problem is that of \emph{outbreak detection}, i.e., detecting the \emph{existence} of a epidemic in a timely manner. In this context, observer placement is a well-studied problem. The optimal solution for timely detection is to place observers at the \emph{k-Median} nodes~\cite{Berry06}; we use this as a  benchmark in Section~\ref{sec:exp_bench}. Furthermore, the optimization of several alternate metrics of interest (e.g., the percentage of infected population at the time of detection) is studied in~\cite{leskovec2007,Krause08}.

\section{The Low-Variance Regime}
\label{sec:low-variance}

In this section, we focus on the low-variance regime. We start by introducing the
setting and the definitions we adopt.

\subsection{Preliminaries}

Let $\G = (V, E)$ be an undirected weighted network. Assuming
$u$ is infected, the weight $w_{uv}$ of edge $uv \in E$ represents the expected
time it takes for $u$ to infect $v$. As the network is undirected, we assume
$w_{uv} = w_{vu}$ for all $uv \in E$.

We assume that the epidemic is initiated by a single unknown source
$s^*$ at an unknown time $t^*$. 
If a node $u$ gets infected at time $t_u$, a
non-infected neighbor $v$ of $u$ will become infected at time $t_v = t_u +
X_{uv}$ where $X_{uv}$ is a random variable with $\E[X_{uv}]=w_{uv}$.

The \emph{time} $t^*$ at which an epidemic starts is unknown. This
adds a significant difficulty to the problem because a \emph{single} observation
is not \emph{per se} informative. 
Instead, we must use the collection of \emph{differences} between observed infection times. If the variance is zero or if it is low compared to edge weights, network distances are
a good proxy for time delays (see Proposition~\ref{Prop:p_err}). We refer to this setting as a \emph{low-variance} regime, as opposed to the \emph{high-variance} regime in which time delays are highly noisy and network distances no longer work as a proxy for time delays.

\textbf{Distance vectors and equivalence between nodes.}
We start with a few definitions. Our setting is similar to~\cite{celis15}.

\begin{definition}[Equivalence]\label{def:equiv}
Let $\mathcal{G}=(V,E)$ and $\mathcal{O} \subseteq V$ with $|\mathcal{O}|=k \geq 2$ be a set of
observers on $\G$. A
node $u$ is said to be equivalent to a node $v$ (which we write $u\sim v$) if and only if, for every $o_i,
o_j \in \mathcal{O}$
\begin{equation}
\label{eq:distinguish}
d(u, o_i) -d(u, o_j) = d(v, o_i) - d(v, o_j).
\end{equation}
where $d(x, y)$ is the (weighted) distance between $x$ and $y$.
\end{definition}

\noindent The relation $\sim$ is reflexive,
symmetric, and transitive, hence it defines an \emph{equivalence relation}. Therefore, a set of observers $\mathcal{O}$ partitions $V$ in
\emph{equivalence classes} (an example is given in Figure~\ref{fig:example_graph}). We denote by $q$
the number of equivalence classes and we let $[u]_\mathcal{O}$ be the class of $u$, i.e., the set of all
nodes that are equivalent to $u$. 

\begin{figure}
\begin{centering}
{\includegraphics[width=.5\columnwidth]{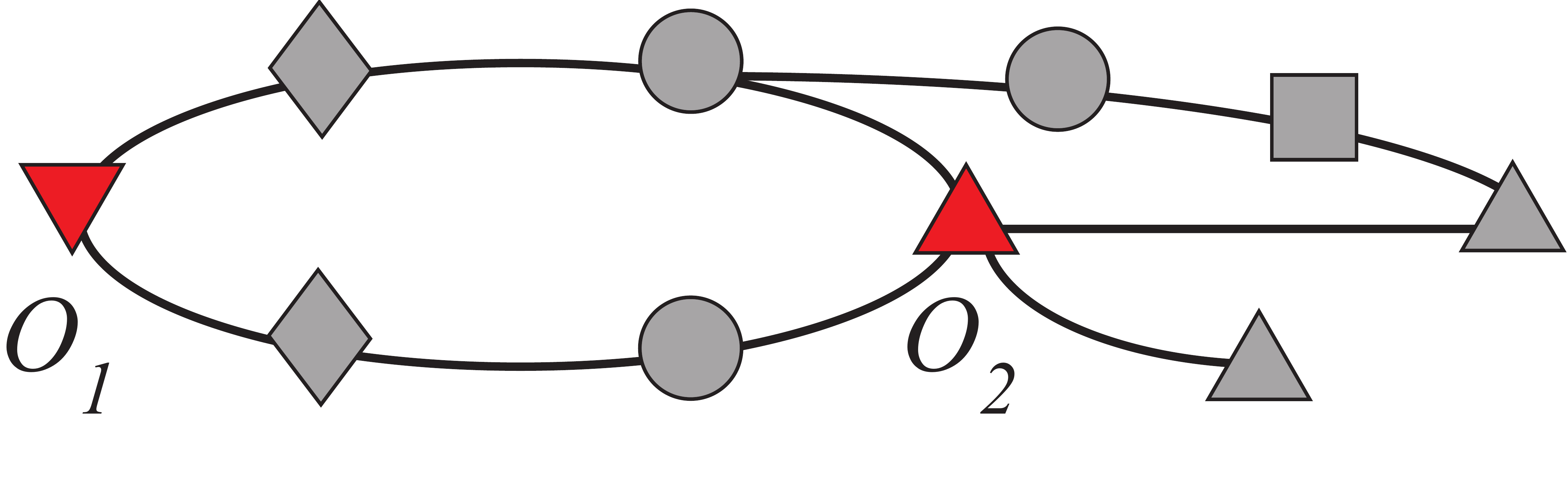}}
\caption{An unweighted network with two observer nodes $o_1$ and $o_2$. Different shapes represent different
equivalence classes, i.e., groups of nodes which are not distinguishable from
the point of view of the observers (solid red). In this example there are $q=5$
equivalence classes.} 
\label{fig:example_graph}
\end{centering}
\end{figure}

When the variance is zero, given an observer set, we can \emph{distinguish} $u$ from $v$ if
Equation~\eqref{eq:distinguish} does \textit{not} hold for $u,v$ and a pair of observers $o_i,o_j$, i.e., if $[u]_\mathcal{O} \neq [v]_\mathcal{O}$.

The problem of finding the minimum-size set of nodes $S$, such that
for every $u,v$ there exist $s_i,s_j \in S$ for which $d(u, s_i) -d(u, s_j) \neq
d(v, s_i) - d(v, s_j)$ is known as the \emph{Double Resolving Set (DRS) Problem}~\cite{Caceres07}. Our problem differs from \emph{DRS} because we focus on the more realistic context in which,
due to limited resources, we want to
allocate a \emph{finite budget} in order to maximize the detection probability (as
opposed to minimizing the number of observers for perfect detection, which is,
in many cases, still prohibitively large). However,
the connection between our problem and \emph{DRS} paves the way for a principled approach. 

We now define a \emph{distance vector} associated with a candidate source,
which, as we will prove in Lemma~\ref{lemma:equiv_DRS}, mathematically captures equivalence in
a manner that is easy to work with.

\begin{definition}[Distance Vector]\label{distance_vector}
Let $\G = (V,E)$ and $\mathcal{O} \subseteq V$ with $|\mathcal{O}|=k \geq 2$ a set of observers
on $\G$.
For each candidate source $s$ the distance vector is  $\mathbf{d}_{s, \mathcal{O}} \in
\mathbb{R}^{k-1}$ with entries $d(s, o_{i+1}) - d(s, o_1)$ for $1 \leq i \leq k-1$.  
\end{definition}

The following lemma, similar in spirit to Lemma 3.1 in~\cite{ChenHW14}, shows
that, the equality between {distance vectors} of 
different candidate sources does not depend on the choice of the \textit{reference
observer} $o_1$.

\begin{lemma}\label{lemma:equiv_DRS} Let $\G = (V,E)$ and $\mathcal{O} \subseteq V$ with
$|\mathcal{O}|=k \geq 2$ and let $u, v \in V$.
Then, $[u]_\mathcal{O}$=$[v]_\mathcal{O}$ if and only if $\mathbf{d}_{u, \mathcal{O}} =
\mathbf{d}_{v, \mathcal{O}}$, independently of the choice of the reference observer
$o_1$ in Definition~\ref{distance_vector}.
\end{lemma}



\textbf{Estimating the source in the low-variance setting.}
We are now ready to describe how we can estimate the source, and define the
probability of correct detection in the zero- and low- variance setting, i.e., when $X_{uv} =
w_{uv}$ a.s.~for every edge $(u, v)$. 

For every observer $o_i \in \mathcal{O}$, denote by $t_i$ the time at which $o_i$ gets infected.
In the zero-variance setting, the observed infection times of nodes $o_2, .., o_K$ with respect to observer $o_1$, i.e., the vector 
$\tau \defeq t_2-t_1, \ldots, t_k-t_1$, is exactly the distance vector of the \emph{unknown} source
$s^*$. Then, if for every $u,v \in V$, $[u]_\mathcal{O} \neq [v]_\mathcal{O}$, the
source can be always correctly identified. We will see in Proposition \ref{Prop:p_err} that this is true also in a
more general \textit{low-variance} framework where we are always able to identify the
equivalence class to which the real source belongs.

We assume a prior probability distribution on the location of the source to be
given, i.e., $Q(u) \defeq \P(s^* = u)$. As we cannot distinguish between
vertices inside $[s^*]_\mathcal{O}$ (otherwise they 
would not be in the same equivalence class), we let our estimated source $\hat
s$ be chosen at random from the conditional probability $Q|_{E^*}(u)\defeq\P(s^*=u|u
\in E^*)$. Hence the
success probability is 
\vspace{-1mm}
\begin{equation}\label{eq:probab_success}
\begin{split}
\Ps(\mathcal{O}) &\defeq  \sum_{s \in V} \P(\hat{s} = s | s^*=s)\P(s^* = s) \\
&= \sum_{s \in V} Q|_{[s]_\mathcal{O}}(s)Q(s) = \sum_{s \in V} \frac{Q(s)}{Q([s]_\mathcal{O})}Q(s),
\end{split}
\end{equation}
\vspace{-1mm}
\noindent and is $1$ if all equivalence classes are singletons.

In the experimental results in
Section~\ref{sec:exp_main} we also look at another relevant metric for the source localization problem, the \textit{expected distance} (weighted or in hops) between the true and estimated source: 
\begin{equation}\label{eq:exp_dist}
\begin{split}
\E[d(s^*, \hat{s})]  &\defeq \sum_{s \in V} \P(s^* = s) \sum_{u \in [s]_\mathcal{O}} \P(\hat s = u | s^* = s)  d(s,u)  \\
&= \sum_{s \in V} \sum_{u \in [s]_\mathcal{O}} \frac{Q(s) Q(u)}{Q([s]_\mathcal{O})} \cdot  d(s,u) .
\end{split}
\end{equation}

\noindent Alternative metrics, including worst-case metrics,
also exist \cite{celis15} (see Appendix~\ref{app:other_metrics} for some examples).

\subsection{Setting}

For ease of exposition, we focus on the case in which the prior distribution on the position of the source is uniform, hence $Q(u)=1/n$ for all $u\in V$.\footnote{Our algorithms and observations can be generalized using Equation~\eqref{eq:probab_success} instead of the simpler formula that we now derive for the uniform case.}

\begin{proposition}
\label{Prop:p_err} 
Let $\G=(V,E)$ be a network of size $n$ and $\mathcal{O} \subseteq V$. Call
$\delta=\min_{u, v: \mathbf{d}_{u, \mathcal{O}} \neq \mathbf{d}_{v, \mathcal{O}}}\|\mathbf{d}_{u,
\mathcal{O}} - \mathbf{d}_{v, \mathcal{O}}\|_\infty$. Assume a uniform prior $Q(u) = 1/n$ for all $u\in V$ and call $D$ the  maximum distance in hops in any shortest path between any node and any observer. 
\begin{enumerate}
\item In the zero-variance case, then ${\Ps(\mathcal{O})= q/n}$, where $q$ is the number of equivalence classes for $\mathcal{O}$;
\item If the transmissions are such that for each $uv \in E$,  ${X_{uv} \in [w_{uv} - \eps, w_{uv}+\eps]}$, we denote as $\Ps^\eps(\mathcal{O})$ the probability of success and we define ${\eps_0=\sup\{\eps>0: \Ps^\eps(\mathcal{O}) = \Ps^0(\mathcal{O})\}}$, we have $\eps_0 > \sfrac{\delta}{2D}$.
\end{enumerate}
 
\end{proposition}

\begin{proof}
\begin{enumerate}
\item By definition,  
$$\Ps(\mathcal{O})=\sum_{[u]_\mathcal{O}}\P(\hat{s} = s^*|s^* \in [u])\P(s^* \in [u]).$$
Hence, 
$$ \Ps(\mathcal{O}) 
=  \sum_{[u]_\mathcal{O}} \frac{1}{|[u]_\mathcal{O}|} \cdot \frac{|[u]_\mathcal{O}|}{n}
= \frac 1 n \sum_{[u]_\mathcal{O} } 1 
= \frac{q}{n}.$$

\item Recall that, for $u, v \in V$, $[u]_\mathcal{O} \neq [v]_\mathcal{O}$ if and only if
$\mathbf{d}_{u, \mathcal{O}} \neq \mathbf{d}_{v, \mathcal{O}}$. Since $\mathbf{d}_{u, \mathcal{O}} \neq
\mathbf{d}_{v, \mathcal{O}}$ implies $\|\mathbf{d}_{u, \mathcal{O}} - \mathbf{d}_{v, \mathcal{O}}\|_\infty
\geq \delta$, if $\eps <\sfrac{\delta}{2 D}$, no estimation error is possible between  $u, v \in V$ such that
$\mathbf{d}_{u, \mathcal{O}} \neq \mathbf{d}_{v, \mathcal{O}}$. Hence $\eps_0 > \sfrac{\delta}{2D}$.
\end{enumerate}
\end{proof}

Note that here $\eps_0$ plays the role of $\sigma_0$ in Figure \ref{fig:scheme}. Indeed, for $\eps<\eps_0$ the variance of the transmission delays does not affect the accuracy of source localization. 

If additional conditions on the weights or on the network topology are
made, more refined bounds for $\eps_0$ can be derived. For example, in a \textit{tree}
with integer weights,
due the uniqueness of the path between two any vertices,
the minimum distance (in the infinity norm) between two distance vectors
 is $2$. Hence, in this case, an accumulated variance of less than $1$ can be tolerated
and we have $\eps_0 > 1/D$.

For the remainder of this section, we will assume ${\eps < \delta/2D}$, which we call the low-variance case. 
Independently of the topology of the network $\G$, the success probability $\Ps$, as well as other
possible metrics of interest, can be computed exactly in polynomial time (see e.g., Equation~\eqref{eq:probab_success} and~\eqref{eq:exp_dist}). In fact, due to Lemma \ref{lemma:equiv_DRS}, it
is enough to compute the distance vector of Definition~\ref{def:equiv} for all
the nodes. Nonetheless, if we have a budget $k \geq 2$ of nodes that we can
choose as observers, finding the configuration that maximizes $\Ps$ is an NP-hard problem. This is a direct consequence of the hardness result of Chen et al.~\cite{ChenHW14}. 

\begin{theorem}
Let $k \geq 2$ be the budget on the number of nodes we can select as observers.
Finding $\mathcal{O} \subseteq V$ such that $\mathcal{O} \in \argmax_{|\mathcal{O}|=k}\Ps(\mathcal{O})$ is NP-hard. 
\end{theorem}

\noindent The proof follows straightforwardly with a reduction
from the \emph{DRS} problem (see Appendix~\ref{app:hardness}).

\subsection{Observer Placement}
Our first main contribution in this paper is a solution to the budgeted observer-placement problem. Our approach, presented in Algorithm~\ref{algo:budget}, is specifically designed for the
source localization problem and has a simple greedy structure: for every node $v \in V$,
initialize $\mathcal{O}\leftarrow \{v\}$ and iteratively add to $\mathcal{O}$ the node $u$
that maximizes the gain with respect to the success probability
until we either run out of budget
or $\Ps=1$. Proposition~\ref{Prop:p_err} ensures that greedily
maximizing the success probability is equivalent to greedily maximizing the
number $q$ of equivalence classes. When adding an element to the observer set,
the partition in equivalence classes can be updated in linear time, hence the
total running time of our algorithm is $O(kn^3)$. Despite bypassing the NP-hardness of the
problem, this might not be sufficiently fast for very large graphs. However, the
procedure is extremely parallelizable and well suited, e.g., for Map-Reduce
(see, for example, the main for loop and the $\mathbf{argmax}$ in the
$\mathbf{while}$ loop).

\begin{algorithm}
\begin{algorithmic}
\caption{(\textsc{lv-Obs}): Observer placement for the low-variance setting.}\label{algo:budget}
\Require{Network $G$, budget $k$}
\For {$v \in V$}
\State $\mathcal{O}_v \leftarrow v$ 
\While{$\Ps(\mathcal{O}_v) \neq 1$ \textbf{and} $\mathcal{O}_v < k$}
    \State $u \leftarrow\argmax_{z \in V \backslash \mathcal{O}_v} [\Ps(\mathcal{O}_v \cup
    \{z\}) - \Ps(\mathcal{O}_v)]$ 
    \State $\mathcal{O}_v \leftarrow \mathcal{O}_v \cup \{u\}$.
\EndWhile
\EndFor
\State \Return $\argmax_{v \in V} \Ps(\mathcal{O}_v)$
\end{algorithmic}
\end{algorithm}

The observer placement obtained through Algorithm~\ref{algo:budget} will be denoted
\textsc{lv-Obs} to emphasize the fact that it has been designed for the case in
which the variance is absent or very small (\textsc{lv} stands for \emph{low-variance} regime).

\subsection{Performance}
As budgeted observer placement (even in the
zero-variance setting) is NP-hard, there is no optimal algorithm to compare
against. Instead, we evaluate the performance of our algorithm against a set of
natural benchmarks that have shown to have good performance in other works
\cite{Seo12, Berry06, Zhang2016} (see Section \ref{sec:exp_bench} for a discussion of
these benchmarks,  Figure~\ref{fig:CR} for the results). 
 

We further compare against two other natural heuristics that also optimize an objective function greedily. 
The first is an adapted version of the  approximation algorithm for the \emph{DRS} problem proposed by Chen at al.~\cite{ChenHW14} and described in Appendix~\ref{app:DRS_CHW}.  
By stopping the greedy process after it selects 
$k$ nodes, we can adapt in a natural way this approximation algorithm and create a heuristic for the budgeted
version. 
The second is a direct minimization of the expected error distance obtained by Equation~\eqref{eq:exp_dist} with $Q(u)=1/n$ for all $u\in V$. Comparing all three approaches, our algorithm outperforms the other two 
(see Appendix~\ref{app:alternate-obj} for details).



\section{The High-Variance Regime}
\label{sec:high-variance}

When the variance is not guaranteed to be low, as defined in Section~\ref{sec:low-variance}, computing analytically the success probability - or other metrics of interest - is unfortunately not possible (except for very simple graphs, like the path in the example of Figure \ref{fig:ex_transition}). Moreover, 
 the estimation of the source is more challenging because the
observed infection delay $t_i - t_j$ can be misleading, especially if the
corresponding observers $o_i$ and $o_j$ are \emph{far} from the source.
Take, for example, a path of length $L$ where the two leaves are the 
only two observers and all edges have weight $1$. Figure~\ref{fig:variance}
shows how the success probability $\Ps$ 
decays faster for increasing values of $L$. Building on this
observation, we propose a strategy for observer placement 
that enforces a controlled distance from a general source node to the 
observer set.

\begin{figure}
\begin{center}
\subfigure[]{\includegraphics[width=0.45\columnwidth]{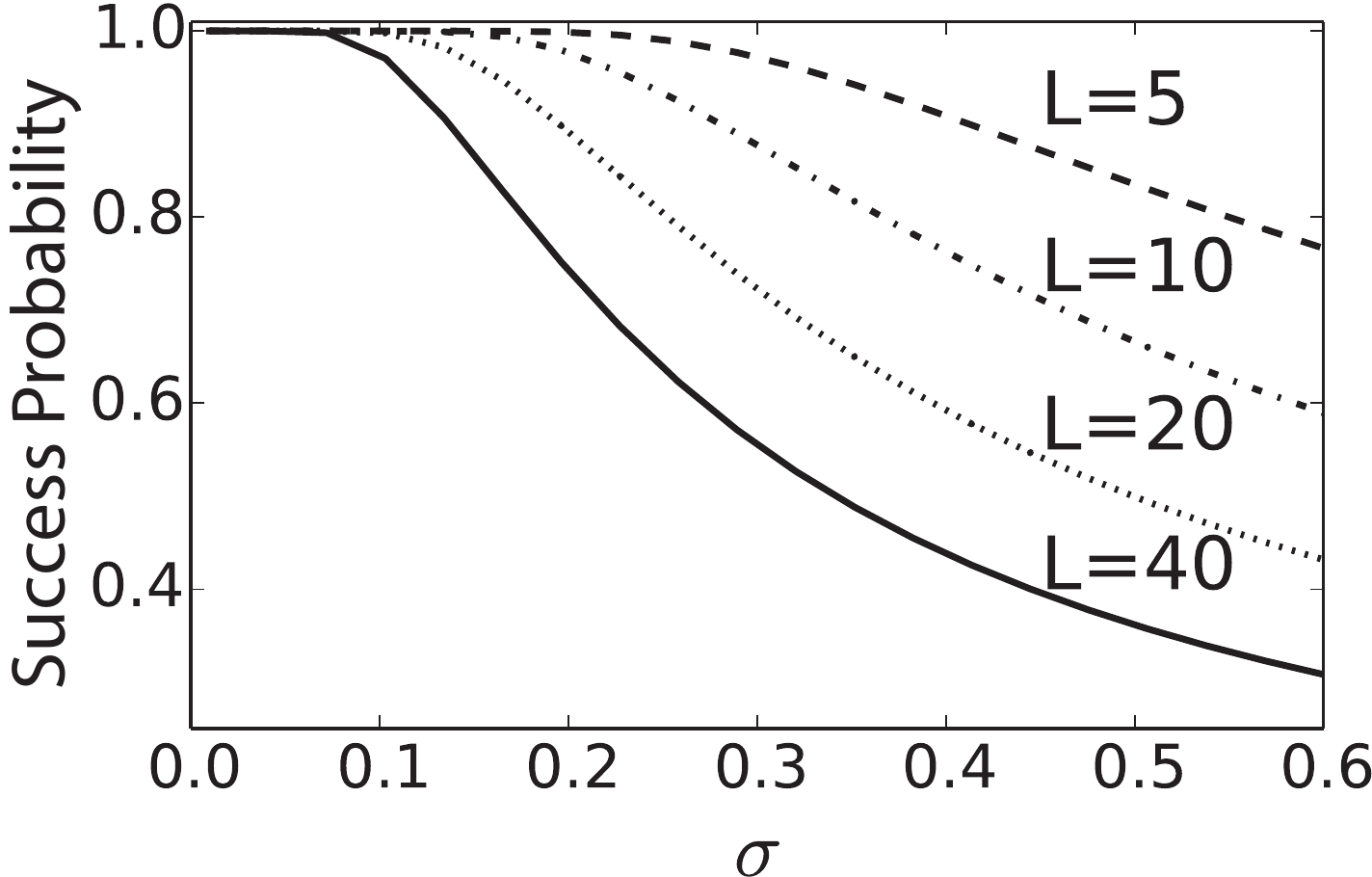}}\label{fig:variance}
\hspace{.5cm}
\subfigure[]{\includegraphics[width=0.3\columnwidth]{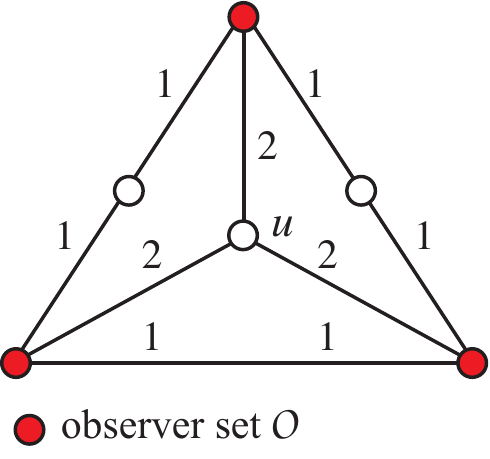}}\label{fig:no-unique-path}
\caption{(a): Success probability $\Ps$ on a path of length $L$ for increasing variance
$\sigma$. (b): Counterexample for the converse of
Lemma~\ref{lem:shortest-path-resolv}; for each pair of observers in $\mathcal{O}$, $u$ is not
contained in the shortest path between them, yet $\mathcal{O}$ is a \emph{DRS}.}
\end{center}
\end{figure}

\subsection{Diffusion Model and Source Estimation}
\label{subsec:noise_model}

For every edge $(u, v)$ the infection delay $X_{uv}$
is distributed as a truncated Gaussian random variable with 
parameters ${(w_{uv}, \sigma w_{uv}, [\sfrac{w_{uv}}{2}, \sfrac{3 w_{uv}}{2}])}$.
More precisely, if ${Y_{uv} \sim \mathcal{N}(w_{uv}, \sigma
w_{uv})}$ is a Gaussian random variable, $X_{uv}$ is obtained by conditioning
$Y_{uv}$ with ${Y_{uv} \in [\sfrac{w_{uv}}{2}, \sfrac{3 w_{uv}}{2}]}$.
This delay distribution has two advantages with respect to the one of
\cite{Pinto12}, i.e., that ${X_{(u,v)} \sim 
\mathcal{N}(w_{uv}, \sigma w_{uv})}$. First, the model admits only strictly
positive infection delays. Second, different values of the standard
deviation $\sigma$
result in different regimes for the propagation, making our model very
versatile. When $\sigma = 0$, $X_{uv}$ boils down to a
deterministic value equal to
the edge weight $w_{uv}$; when $\sigma$ is large, the distribution of $X_{uv}$
becomes closer to uniform ${U([\sfrac{w_{uv}}{2}, \sfrac{3 w_{uv}}{2}])}$.
Finally, when $\sigma$ is strictly positive but small, 
${X_{uv} \approx \mathcal{N}(w_{uv}, (\sigma w_{uv})^2)}$.
In Appendix~\ref{app:estimator}, we explain how an approximated
maximum likelihood estimator for the source can be derived in this setting.

\subsection{Observer Placement}

First, we formalize why distances between observers are important: 
If $o_i, o_j$ are two observers and the source is ${s^*
\in \mathcal{P}(o_i, o_j)}$, then 
\begin{equation}\label{eq:path_variance}
\var(t_i - t_j) \approx \sigma^2 \left[ \sum_{(uv) \in \Path(o_i, o_j)}w_{uv}^2
\right]
\end{equation}
where $\Path(x,y)$ denotes the shortest path from $x$ to $y$, written as a
sequence of edges. 
Although we cannot control $\sigma$, we can  control the \emph{path length}
between observers.\footnote{A relevant but orthogonal line of
work would study how to control the parameter $\sigma$ by, e.g.,
immunizations, quarantines, or other preventative measures and is outside the
scope of our work.}
We make use of the following sufficient condition for a set to be a $\emph{DRS}$, 
i.e., for an observer set to guarantee optimal source detection. 

\begin{lemma}\label{lem:shortest-path-resolv}
Let $\G=(V, E)$ be a network, $\mathcal{O} \subseteq V$. If for every $u\in V$ there
exist $o_1, o_2 \in \mathcal{O}$ such that there is a unique shortest path
$\mathcal{P}(o_1, o_2)$
between $o_1$ and $o_2$ and $u \in \mathcal{P}(o_1, o_2)$, then $\mathcal{O}$
is a \emph{DRS} for $G$.
\end{lemma}

\begin{proof}
Let $u, v \in V \backslash \mathcal{O}$. We will prove that there exist $o_1, o_2 \in \mathcal{O}$ such
that the pair $(u, v)$ is resolved by $(o_1, o_2)$, i.e., $ d(v, o_1)-d(u,o_1)
\neq d(v, o_2)- d(u, o_2)$. Let $o_1, o_2 \in \mathcal{O}$ such that
$u$ appears in the unique shortest path $\mathcal{P}(o_1, o_2)$ and $o_3, o_4 \in S$ such that $v$ appears in the
unique shortest path $\mathcal{P}(o_3, o_4)$. If $v \in \mathcal{P}(o_1, o_2)$
or $u \in \mathcal{P}(o_3, o_4)$ than $u$ and $v$ are resolved by, respectively,
$(o_1, o_2)$ or $(o_3,
o_4)$. Take $v \notin \mathcal{P}(o_1, o_2)$
and $u \notin \mathcal{P}(o_3, o_4)$. In this case, $\{o_1,
o_2\} \neq \{o_3, o_4\}$. Let
us suppose without loss of generality that $o_1 \notin \{o_3, o_4\}$. 
We look only at the case where $(o_1, o_2)$ does not resolve $(u, v)$ and
prove that the pair is indeed resolved by two vertices in $\mathcal{O}$. Since $(o_1,
o_2)$ does not resolve $(u, v)$,
there exists $c\in \R$ such that $d(v, o_1)-d(u,o_1) = c = d(v, o_2)- d(u,
o_2)$. Since the unique shortest path between $o_1$ and $o_2$ goes through $u$
we have that $c>0$. We prove that either $(o_1, o_3)$ or $(o_1, o_4)$ resolves
$(u, v)$. If this was not the case, we would have the following equalities:
\begin{align*}
c &= d(v, o_1)-d(u,o_1) = d(v, o_3)- d(u, o_3)\\
c &= d(v, o_1)-d(u,o_1) = d(v, o_4)- d(u, o_4).
\end{align*}
Since $c>0$, $d(v, o_3) > d(u, o_3)$ and $d(v, o_4) > d(u, o_4)$ giving a
contradiction with $v$ (and not $u$) being on the shortest path $\mathcal{P}(o_3, o_4)$. We conclude
that $(u, v)$ are resolved by either $(o_1, o_3)$ or $(o_1, o_4)$.
\end{proof}
The converse of this lemma is not true: If $\mathcal{O}$ double resolves $\G$, it
is not even true that for every node $u$ there must exist $o_1, o_2
\in \mathcal{O}$ such that $u$ is contained in \emph{some} shortest path between $o_1$ and
$o_2$ of (see the Example in Figure~\ref{fig:no-unique-path}).

\textbf{Path covering strategy.} We take Lemma~\ref{lem:shortest-path-resolv} as a basis for deriving a \emph{path
covering} strategy for observer placement. In practice, the condition about the \emph{uniqueness} of the shortest path is too strong and excludes many potentially useful observer nodes\footnote{Experimentally we see that in many practical situations two shortest paths differ only by a few nodes and the majority of nodes on the path are resolved by the two extreme nodes.}. 
This is why we relax the condition of Lemma~\ref{lem:shortest-path-resolv} and we prefer, when the shortest path is not unique, to select one arbitrarily. 
Let $S\subseteq V$ be a set of observers and $L$ a positive integer: We call $P_{L}(S)$ the set of nodes that lie on a shortest path of length at most $L$ between any two observers in the set $S$. 
Given a budget $k$, and a positive integer $L$, we denote by $S^*_{k, L}$ the set of $k$ vertices that maximize the cardinality of $P_{L}(S)$. 
We call $L$ the \emph{length constraint} for the observer placement because we consider an observer to be \emph{useful} for source localization only if it is within distance $L$ from another observer. 
$S^*_{k, L}$ can be approximated greedily as in Algorithm~\ref{algo:path_covering}.\footnote{The running time of Algorithm \ref{algo:path_covering} is $O(n^2k^2)$, however, as with the low-variance case, this is highly parallelizable and hence tractable even for large networks.}

\begin{algorithm}
\begin{algorithmic}
\caption{(\textsc{hv-Obs}): Observer placement for the high-variance setting.}\label{algo:path_covering}
\Require{Network $G(V, E)$, budget $k$, length constraint $L$}
\State $n \leftarrow |G|$
\For {$v \in V$}
\State $\mathcal{O}_v \leftarrow v$ 
\While{$|P_{L}(\mathcal{O}_v)| \neq n$ \textbf{and} $\mathcal{O}_v < k$}
    \State $u \leftarrow\argmax_{z \in V \backslash \mathcal{O}_v} [|P_{L}(\mathcal{O}_v \cup
    \{z\})| - |P_{L}(\mathcal{O}_v)|]$ 
    \State $\mathcal{O}_v \leftarrow \mathcal{O}_v \cup \{u\}$.
\EndWhile
\EndFor
\State \Return $\argmax_{v \in V} |P_{L}(\mathcal{O}_v)|$
\end{algorithmic}
\end{algorithm}


We will refer to the observer placement produced by
Algorithm~\ref{algo:path_covering} as \textsc{hv-Obs}$(L)$ to emphasize that it is designed for the high-variance
case.

\textbf{Comparison with Algorithm~\ref{algo:budget}.} 
Note that taking $L$ equal to the maximum weighted distance $\Delta$ does not
make Algorithm~\ref{algo:path_covering} equivalent to Algorithm~\ref{algo:budget},
i.e., we do not obtain \textsc{lv-Obs}. 
To see how the two algorithms could give different results, take a cycle of odd length $d$ with a leaf node $\ell$ added as a neighbor to an
arbitrary node $v$ and assume to start the algorithm with initial set $\{v\}$.
At the first step, the two algorithms will make the same choice, choosing one of
the two nodes that is at distance $(d-1)/2$ from $v$. At the second step
however, $\textsc{lv-Obs}$ will add $\ell$ (a DRS contains all
leaves \cite{ChenHW14}), whereas Algorithm~\ref{algo:path_covering} will add a
node on the cycle. This observation is key to our results because it explains
why Algorithm \ref{algo:path_covering} results in a more uniform (and hence
\textit{variance-resistant}) observer placement with respect to \textsc{lv-Obs}. \textsc{hv-Obs} operates a
trade-off between the average distance to the observers and the maximization of $\Ps$.

\textbf{Choice of the L parameter.} 
How could one optimally set $L$? Needless to say, the optimal 
$L$ depends on the network topology and on the available budget: Clearly, for
a larger budget a smaller $L$ is preferred.  

The cardinality of $P_{L}(\mathcal{O})$ is a good
proxy for the performance of $\mathcal{O}$. The value
$|P_{L}|$ is increasing in $L$ and reaches its maximum for
$L$ equal to the maximum weighted distance ($L=\Delta$).
For small $L$, $|P_{L}(\textsc{hv-Obs})| <
|P_{\Delta}(\textsc{lv-Obs})|$ but for $L$ large enough this is no more the
case. See Figure~\ref{fig:covered} for an example.
Our empirical results 
suggest that $L$ should be chosen as the maximum for which $|P_{L}(\textsc{hv-Obs})| \leq
|P_{\Delta}(\textsc{lv-Obs})|$. The key property of 
\textsc{hv-Obs} with respect to \textsc{lv-Obs} is that observers are spread
more \emph{uniformly} without \emph{losing} too much in terms of
success probability $\Ps$: Figure~\ref{fig:covered-prob-err} shows
$|P_{L}(\textsc{hv-Obs})|$ and $\Ps$ as a function of $L$.

\textsc{lv-Obs} and \textsc{hv-Obs} can give drastically different observers (see Appendix~\ref{app:figures} for an example).

\begin{figure}
\begin{center}
\includegraphics[width=0.5\columnwidth]{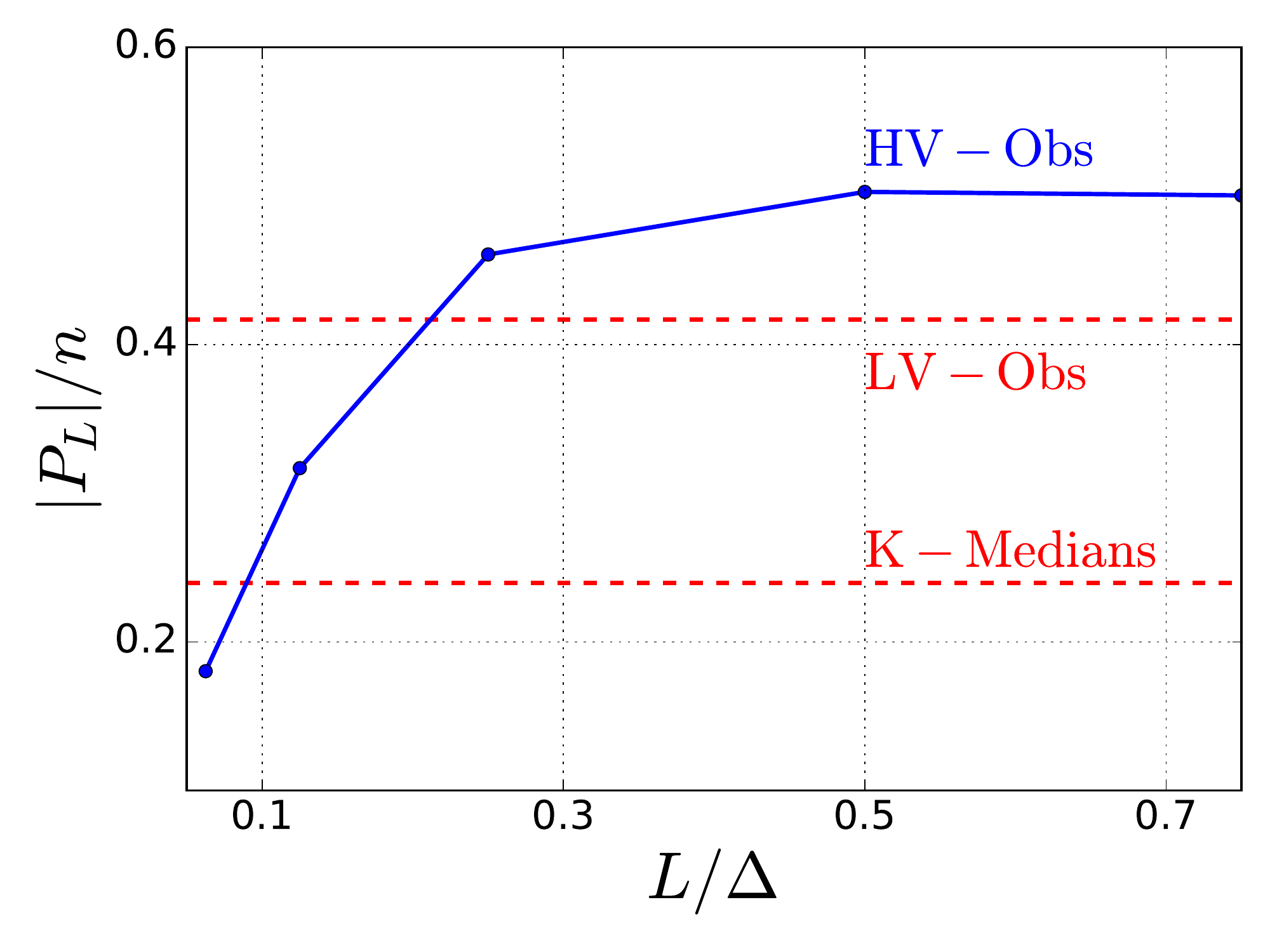}
\caption{Fraction of nodes in $P_{L}(\cdot)$ for the California dataset with
$2\%$ of observers.}\label{fig:covered}
\end{center}
\end{figure}

\begin{figure}
\centering
\subfigure[CR]{\includegraphics[width=0.47\columnwidth]{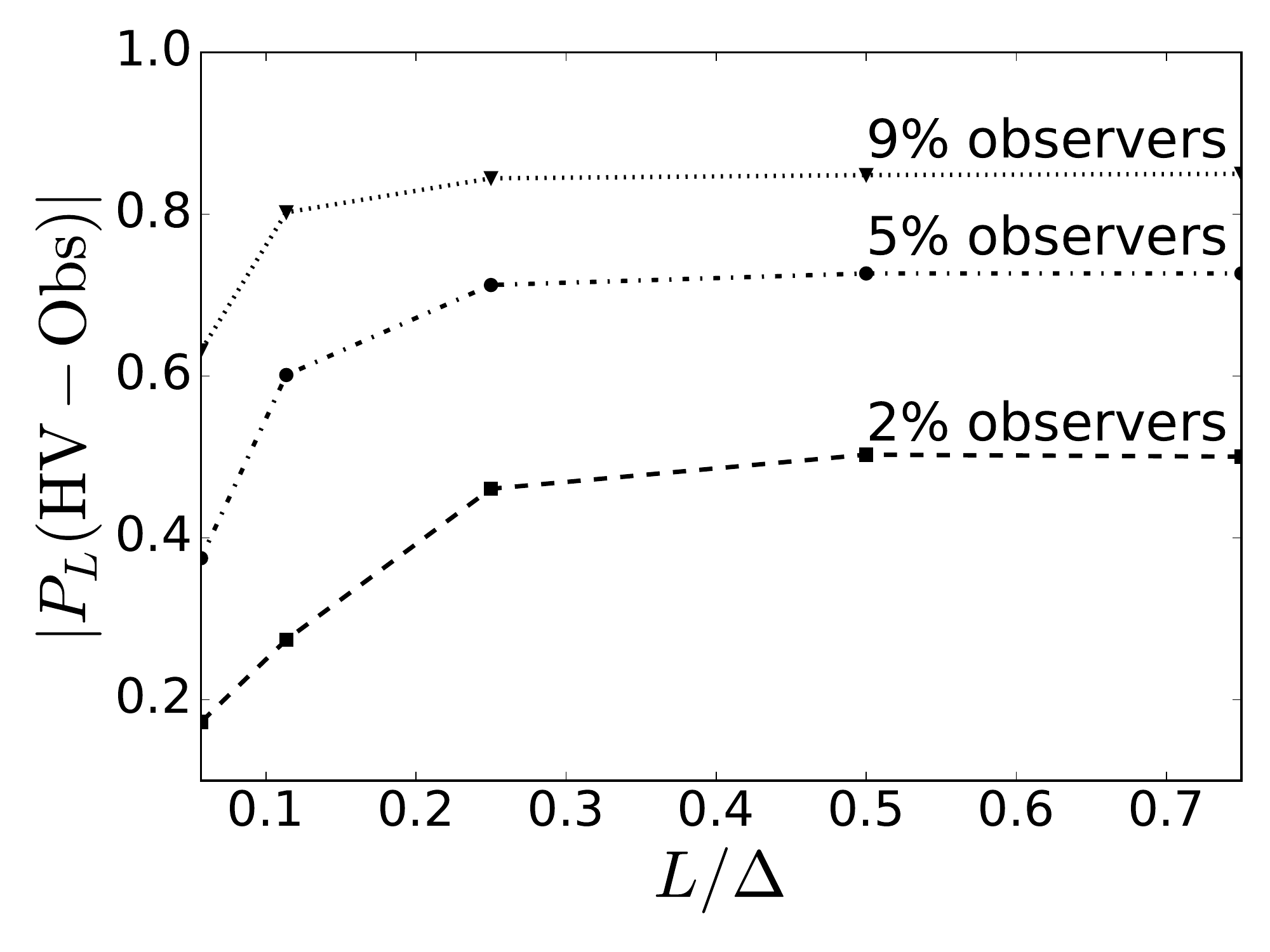}}
\subfigure[CR]{\includegraphics[width=0.47\columnwidth]{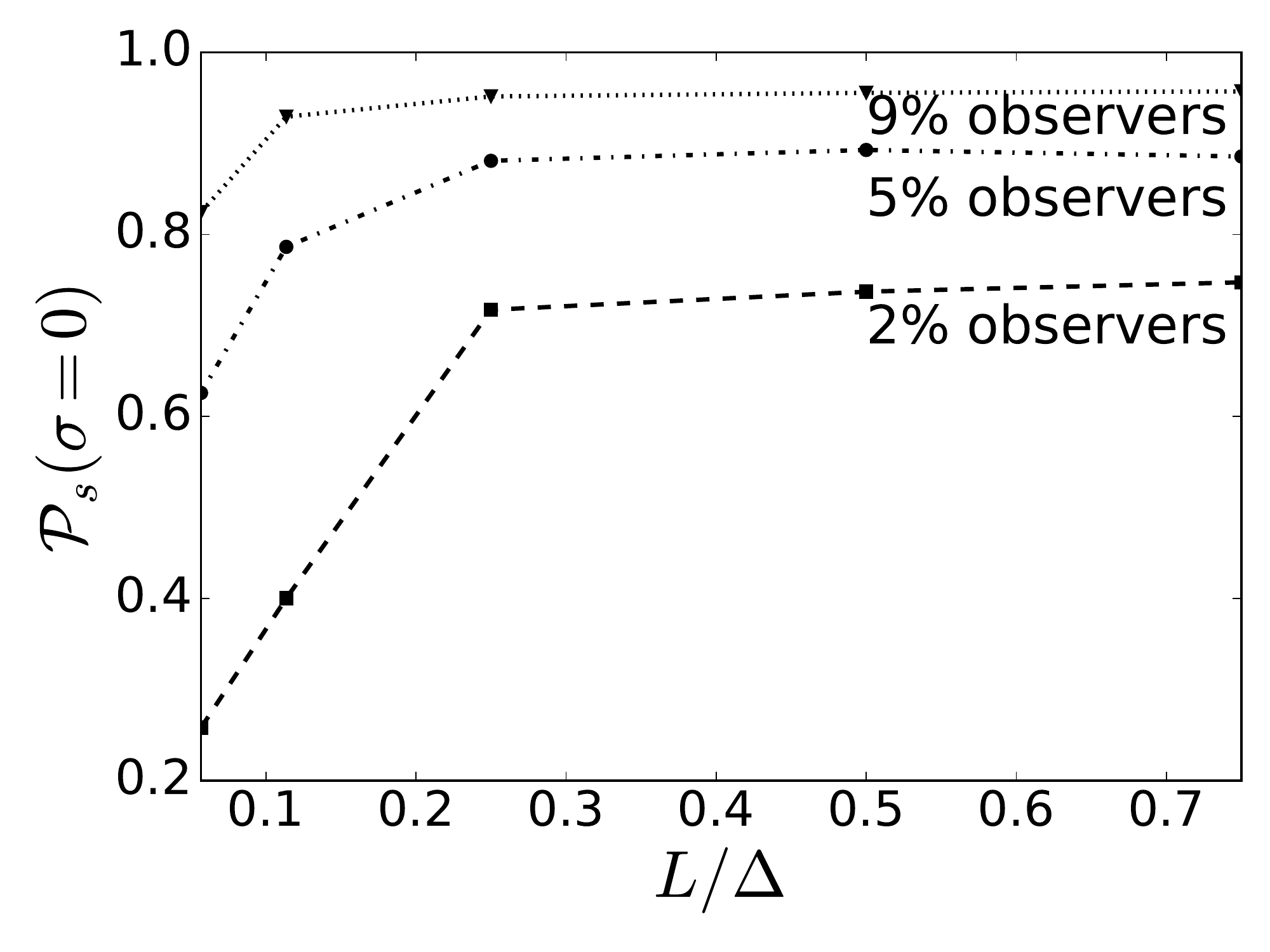}}
\subfigure[F \& F]{\includegraphics[width=0.47\columnwidth]{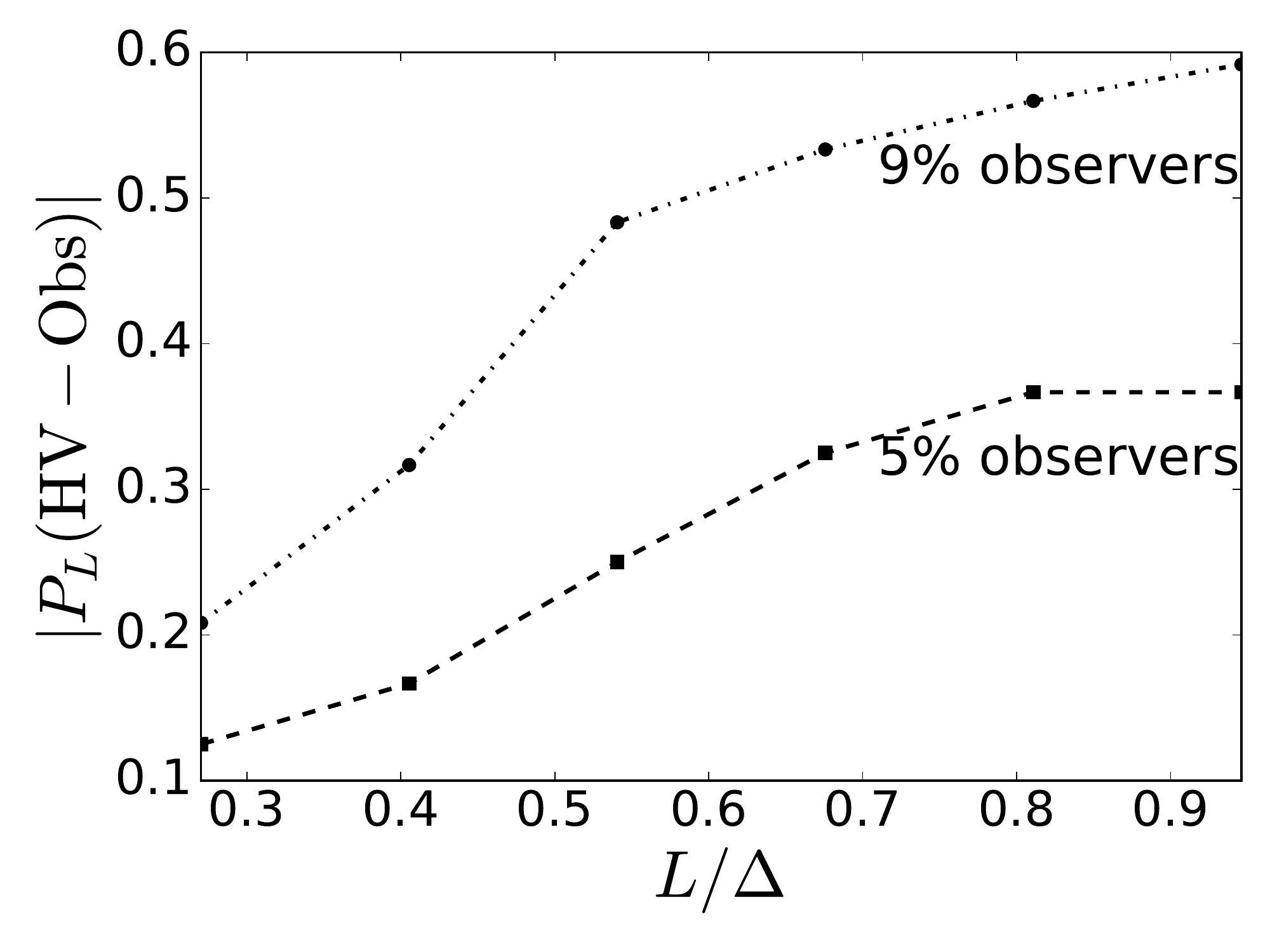}}
\subfigure[F \& F]{\includegraphics[width=0.47\columnwidth]{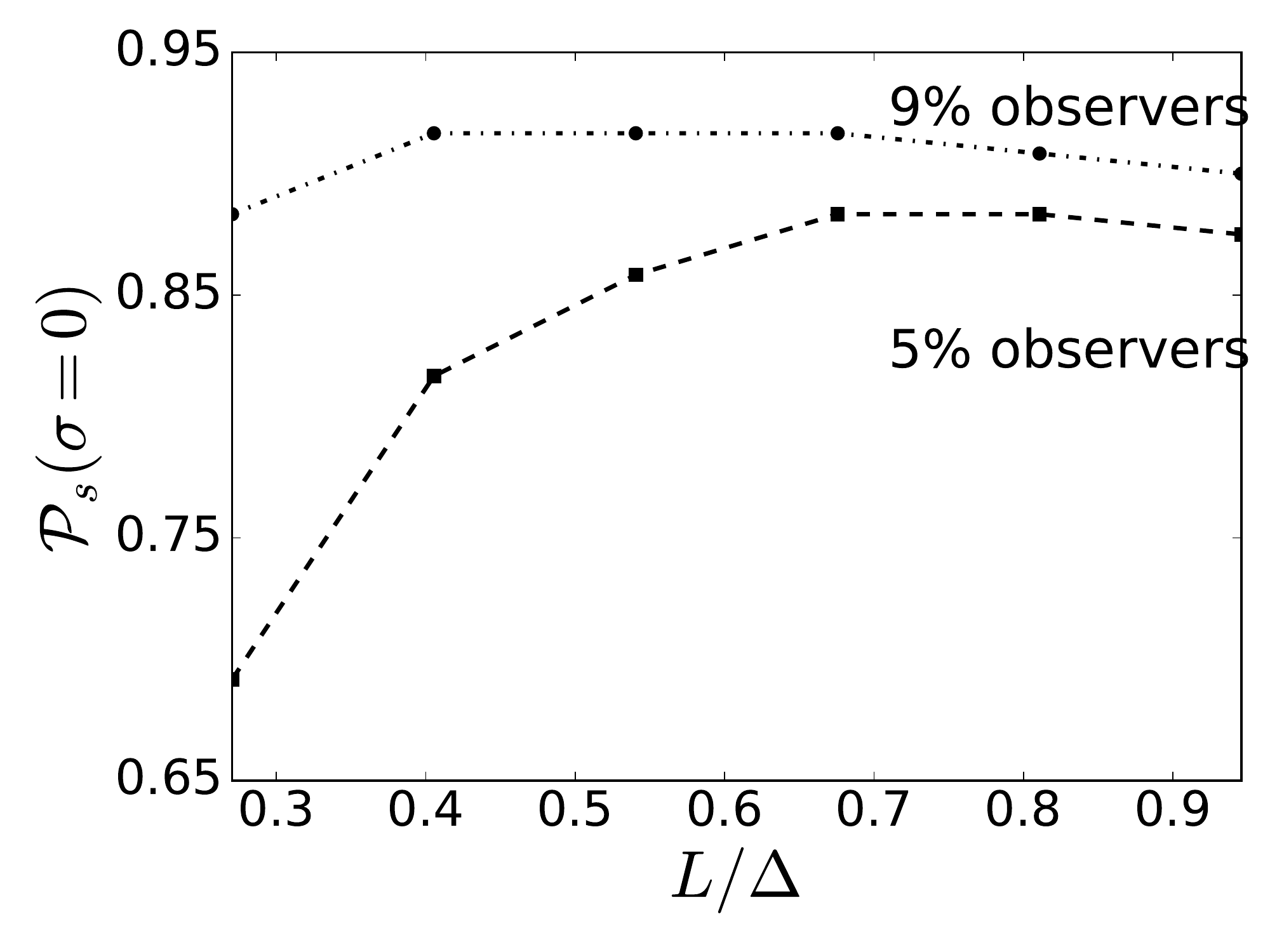}}
\caption{Fraction of nodes in $P_{L}(\textsc{hv-Obs})$ and success probability as a function of
$L/\Delta$ for the CR and the F \& F datasets comparing with the zero-variance setting.}
\label{fig:covered-prob-err}
\end{figure}

\section{Empirical results} \label{sec:exp_main}

We purposely run our experiments on three very different real-world networks that, in
addition to being relevant examples of networks for epidemic spread, display
different characteristics in terms of size, diameter, clustering coefficient
and average degree (see Table~\ref{tab:graphs}), enabling us to test the performance of our methods on
various topologies.


\subsection{Datasets}
\noindent The three networks we consider are:
\begin{itemize}
\item[-] Friend \& Families (F \& F). This is a dataset containing 
phone calls, SMS exchanges and bluetooth proximity, among a
community living in the proximity of a university campus~\cite{aharony2011}. We select the largest connected component of individuals who took part in the experiment during its whole duration. The edges are weighted, according to the number of phone calls, SMSs, and bluetooth contacts.

\item[-] Facebook-like Message Exchange (FB)~\cite{opsahl2009}. As the individuals included in this dataset were living on the same university-campus, the number of messages exchanged is likely to be a good measure 
of in-person interaction. We selected links on which at least one message
was sent in both directions and individuals that had more than $1$ contact. 

\item[-] California Road Network (CR)~\cite{cal_data}. In order to obtain a single connected component and remove points that effectively represent the same location, we collapsed the points falling within a distance of $2$ km. Moreover we iteratively deleted all leaves.\footnote{The roads that cross the state border are not completely tracked in this dataset and terminate with a leaf. Some other leaves might represent remote locations, not necessarily close to the borders, but their influence on the epidemic should anyway be very low.} 
The diameter of this network is very large compared with that of the other two networks. The edges are weighted according to a rescaled version of the real distance (measured in km). 
\end{itemize}

\noindent In all three networks, edges are given (non-unit) integer \textit{weights}, which is
 realistic in many applications as the expected transmission delays are known only up to some level of precision.
Integer weights do \emph{not} simplify the estimation of the source; in fact, this makes it \emph{more} difficult to
distinguish between vertices.
For example, if the edges of the CR network were weighted according to
the Euclidean distance between the two endpoints, \textsc{lv-Obs} would use only a very small portion of the budget and the comparison would not be meaningful.

\begin{table*}[t]
\centering

\begin{tabular}{l| ccccccccc } 
\hline
& $|V|$ & $|E|$  & $\min(w_{uv})$ & $\mathrm{avg}(w_{uv})$ & $\max(w_{uv})$ & Avg Degree & Diameter & Avg Dist & Avg Clust.\\
\hline 
Friends \& Families &120  & 563  & 4 & 5.58 & 7 &  9.38 &  6 & 17.5 & 0.67\\
Facebook Messages   &1020 & 6205 & 1 & 2.97 & 5 & 12.16 &  5 & 6.69 & 0.09\\
California Roads    &1259 & 1801 & 1 & 1.71 & 9 &  2.86 &  66 & 55.3 & 0.2\\

\hline
\end{tabular}

\caption{Displays statistics for networks examined}\label{tab:graphs}
\end{table*}


\begin{figure*}
\centering
\subfigure[CR, 2\%
observers]{\includegraphics[width=0.6\columnwidth]{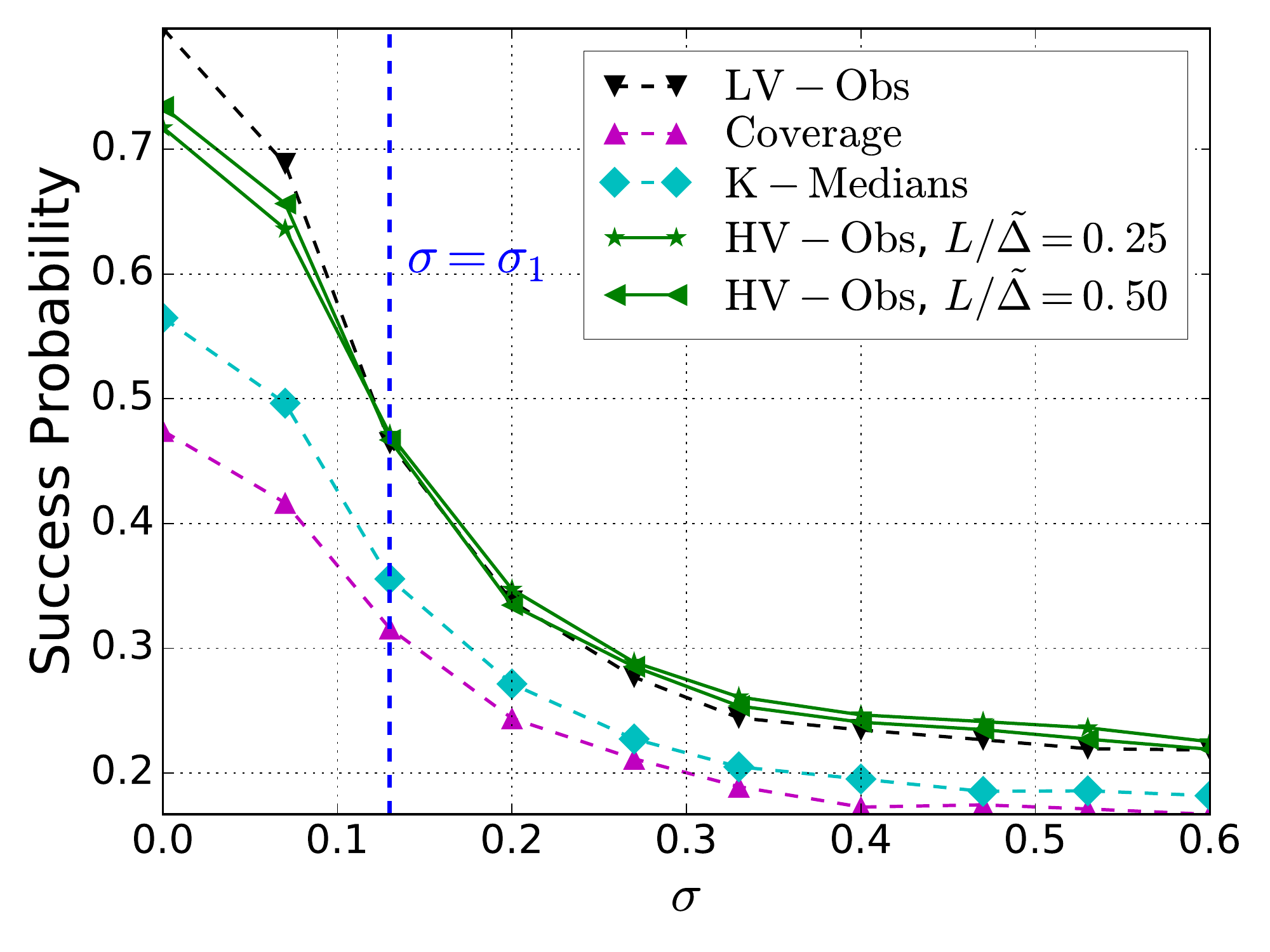}}
\subfigure[CR, 5\%
observers]{\includegraphics[width=0.6\columnwidth]{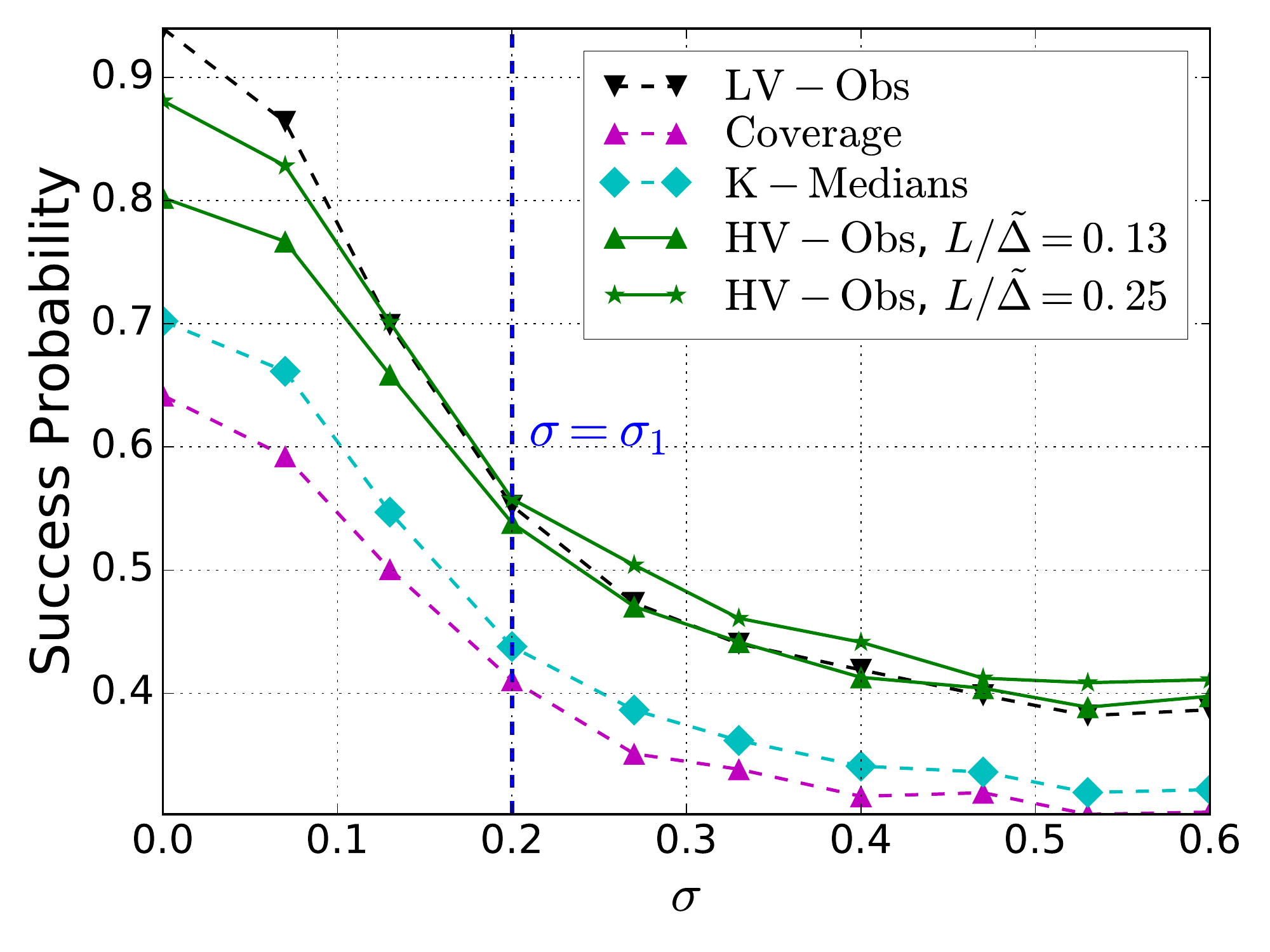}}
\subfigure[CR, 9\%
observers]{\includegraphics[width=0.6\columnwidth]{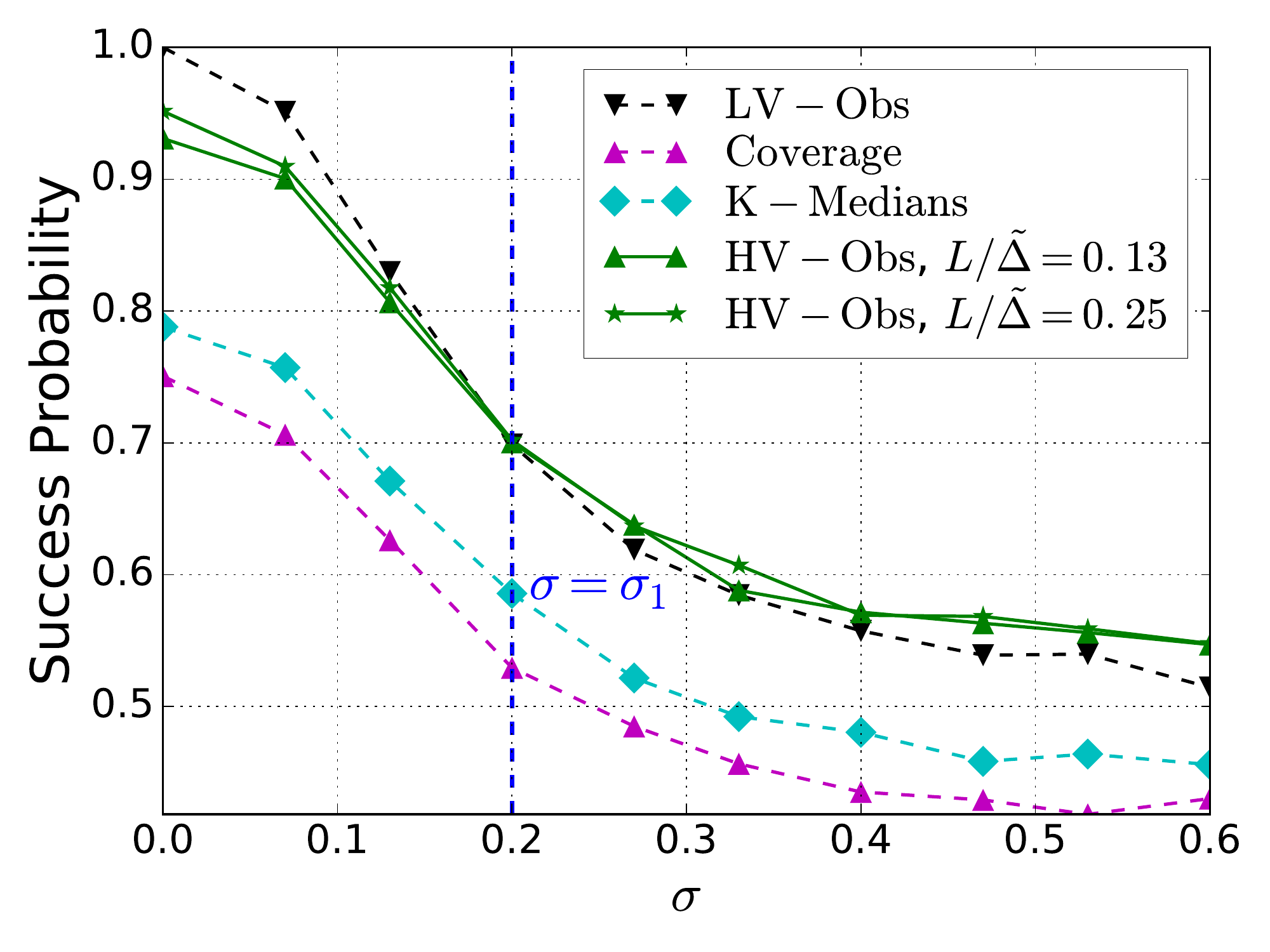}}\\
\subfigure[FB, 5\%
observers]{\includegraphics[width=0.6\columnwidth]{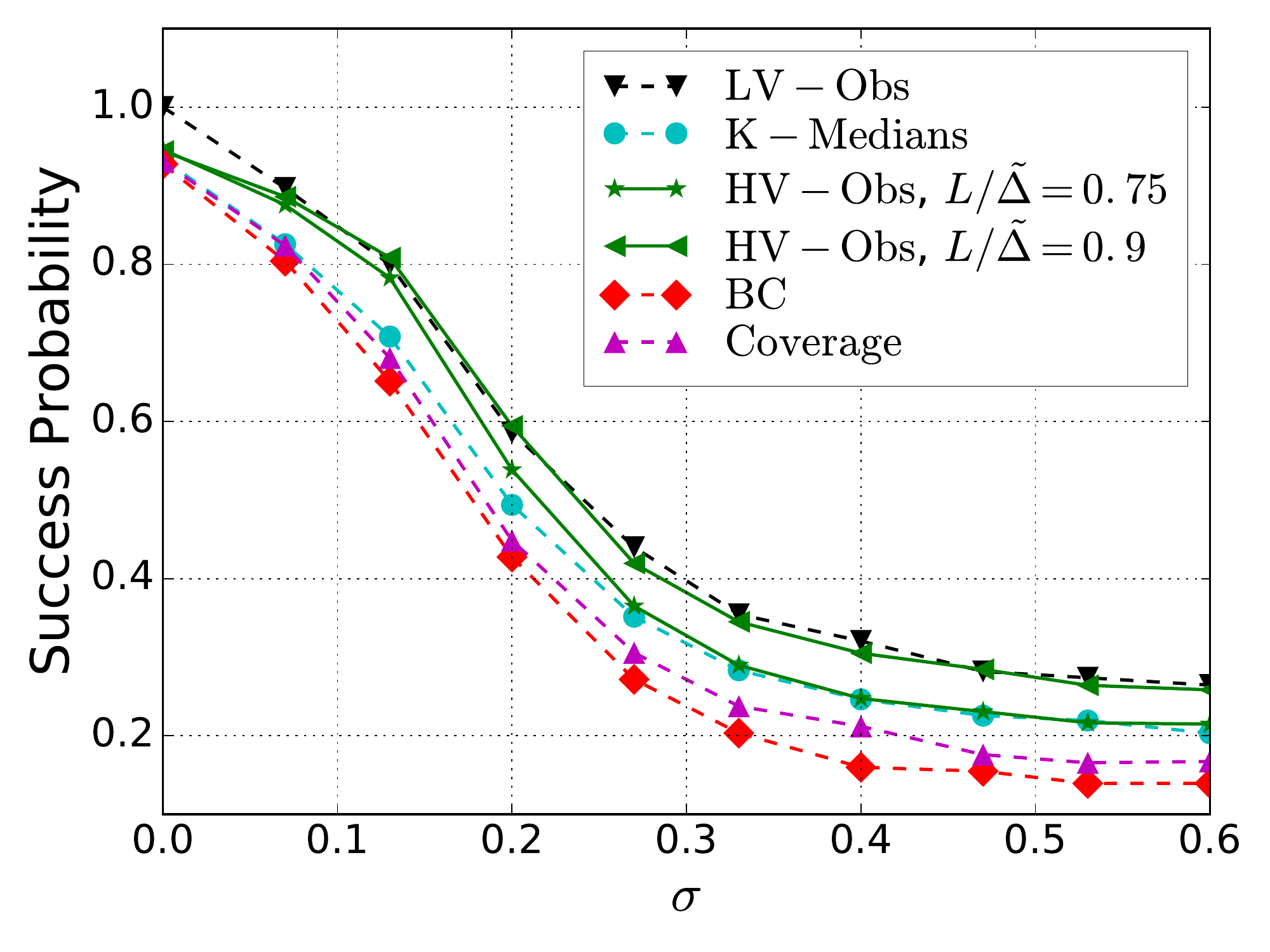}}
\subfigure[F \& F, 5\%
observers]{\includegraphics[width=0.6\columnwidth]{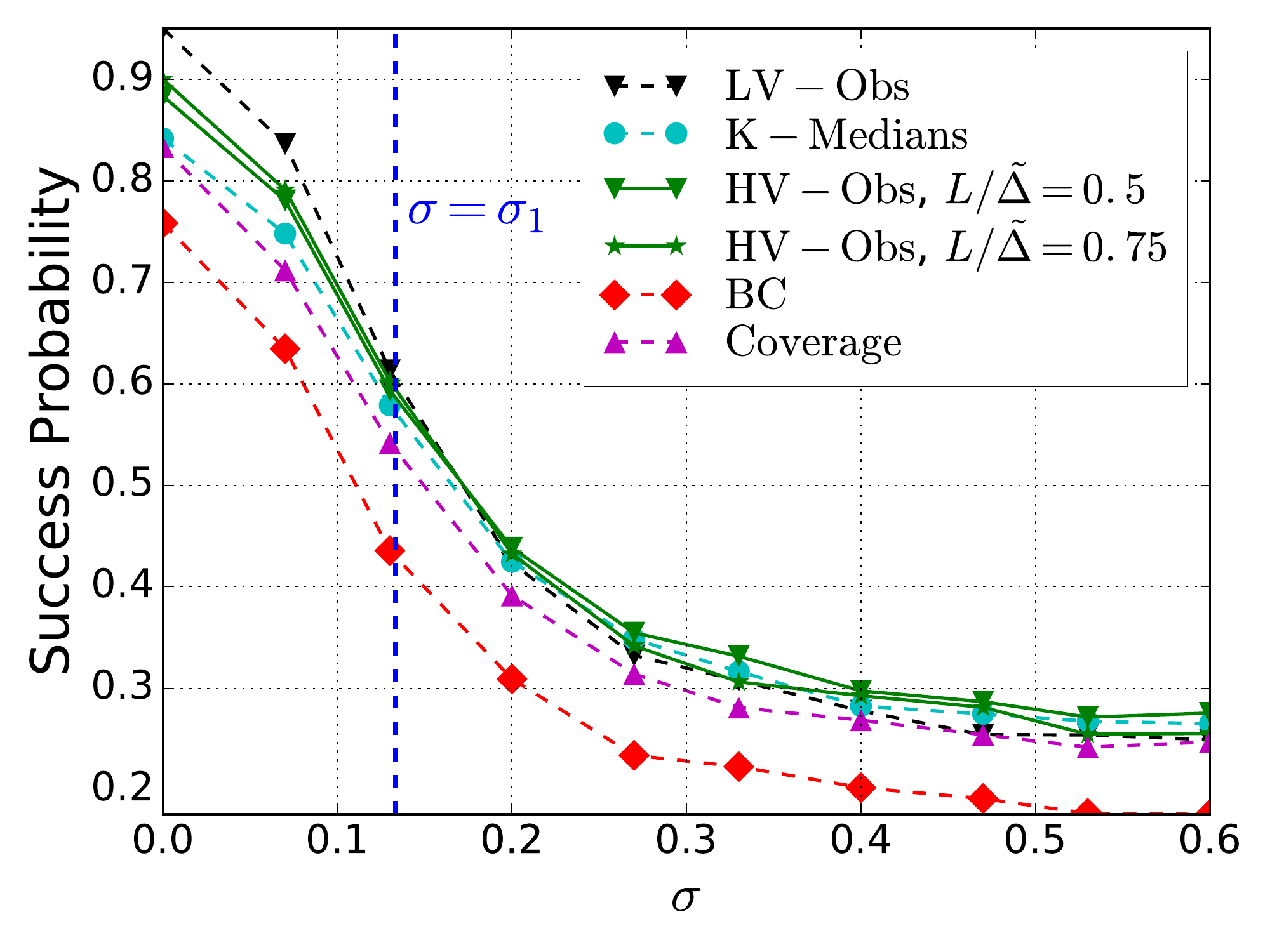}}
\subfigure[F \& F, 10\%
observers]{\includegraphics[width=0.6\columnwidth]{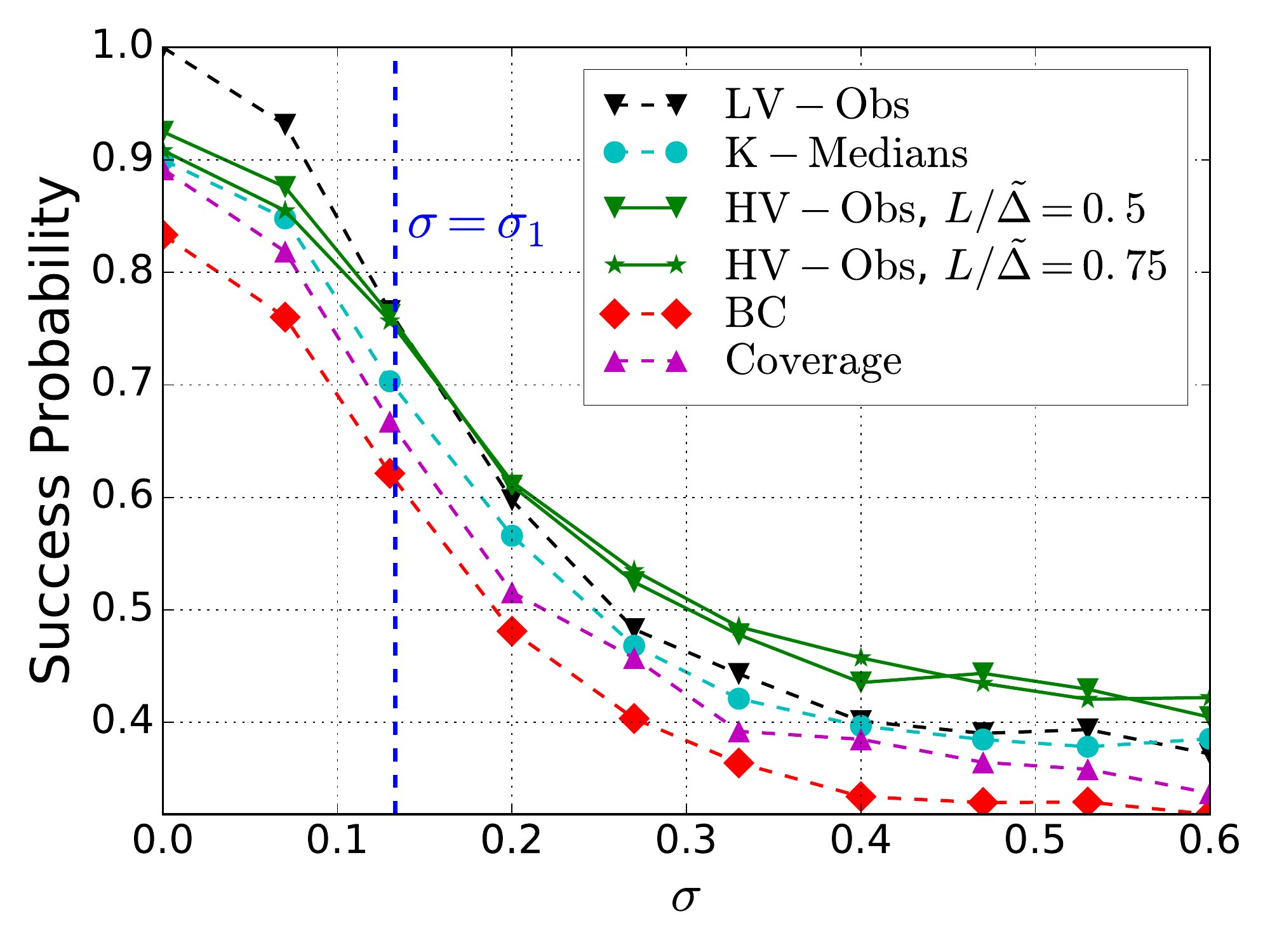}}
\caption{Success probability $\Ps$ as variance
$\sigma$ is increased.}
\label{fig:CR}
\label{fig:ff}
\end{figure*}
\subsection{Comparison against Benchmarks}
\label{sec:exp_bench}

\noindent We compare \textsc{hv-Obs} against the following benchmarks:
\begin{enumerate}
\item \textsc{lv-Obs}: this is our solution for the
low-variance case (see Section~\ref{sec:low-variance}).

\item \textsc{BC} (Betweenness Centrality): This is a popular
method for placing observers for source-localization  
(see, e.g., \cite{Louni14} and \cite{Seo12}, where it emerges as the best
heuristic for observer placement among those tested). 
It consists of the $k$ nodes having the largest BC, which is defined, for all $u \in V$ as 
$$BC(u) = \sum_{x, y \in V, x \neq y} \frac{\sigma_{x,y}(u)}{\sigma_{x,y}}$$
where $\sigma_{x,y}$ is the number of shortest paths between $x$ and $y$ and
$\sigma_{x,y}(u)$ is the number of those paths that passes through $u$.

\item Coverage-rate (\textsc{Coverage})~\cite{Zhang2016}: This approach maximizes 
the number of nodes that have an observer as a neighbor, i.e.,
$$\mathcal{C}(\mathcal{O}) = |\cup_{o \in \mathcal{O}} N_{o}|/n,$$
where $N_{o}$ denotes the set of neighbors of $o$ and $n=|V|$.
It has been shown to outperform several heuristics with a
diffusion model and an estimation setting that are very similar to ours. 

\item \textsc{k-Median}: this is the optimal placement for
the closely-related problem of maximizing the detectability of a flow \cite{Berry06}. The
\textsc{k-Median}
placement is the set of $k$ nodes $\mathcal{O}$ such that
\[\mathcal{O} = \argmin_{|\mathcal{O}|=k} \sum_{s \in V} (\min_{o \in \mathcal{O}} d(s, o)).\]
Determining the \textsc{k-Medians} of a network is 
NP-hard \cite{Kariv79}; we use a greedy heuristic for \textsc{k-Medians}.
\end{enumerate}

\subsection{Experimental Results}\label{sec:results}

We estimate $\Ps$ and $\E[d(s^*, \hat{s})]$ for
different values of the variance $\sigma$. We generate epidemics by using each node in turn as the source. For the FB and CR
datasets, we run $5$ simulations per node and variance level; and for the F \& F dataset, as
the network is smaller, we run $20$ simulations per node and variance
level. For the FB and CR datasets, we estimate the
source based on the first $20$ observations only: Given the large size of the
network, it would be unrealistic to wait for all the network to get
infected before running the algorithm. The results for
$\Ps$ are displayed in Figure~\ref{fig:CR}. An approximation of the value $\sigma_1$, above which
\textsc{hv-Obs} outperforms \textsc{lv-Obs}, is marked with a vertical line. 
For the expected distance (weighted and in hops), see Appendix~\ref{app:figures}. 

We first take as budget for the observers the minimum
budget for which $\Ps($\textsc{lv-Obs}$) = 1$. This 
corresponds to $k \sim 9 \%$ for the F \& F dataset, $k \sim 9 \%$ for the
CR network and $k \sim 5\%$ for the FB dataset. This is the setting in which we expect 
the improvement of \textsc{hv-Obs} over \textsc{lv-Obs} to be especially strong: For smaller values of $k$ we expect
\textsc{lv-Obs} to be nearly optimal even in the high-variance regime because we do
not have enough budget to contrast both the topological
\emph{undistinguishability} among nodes (what \textsc{lv-Obs} is designed for)
and the accumulation of variance (what \textsc{hv-Obs} is designed for). For the
F \& F and the CR networks, we also experiment with smaller
percentages of observers and consistently find an improvement of
\textsc{hv-Obs} over \textsc{lv-Obs} in the high-variance regime: Below a certain amount of variance $\sigma_1$ \textsc{lv-Obs} performs better than \textsc{hv-Obs} for any choice of the parameter $L$, whereas above
$\sigma_1$ a calibrated choice of $L$ leads to a significant improvement. Such $L$ stays constant for all $\sigma>\sigma_1$, i.e., with the notation of Figure~\ref{fig:scheme} we have $\sigma_1=\sigma_F$.  
For the FB dataset instead, probably due to the low
diameter with respect to the number of nodes, we observe that \textsc{hv-Obs} does not improve
on \textsc{lv-Obs} for any value of $L$.  Both \textsc{lv-Obs}
and \textsc{hv-Obs} systematically outperform the baseline heuristics for
observer placement that we described in Section~\ref{sec:exp_bench}. For the
CR dataset the performance of Betweenness Centrality is particularly
poor and the results are not shown. The Coverage Rate heuristic outperforms
Betweenness Centrality on all three networks (confirming what found by
by Zhang et al.~\cite{Zhang2016}) but is consistently less effective than
K-Medians and our methods.

\subsection{Robustness}
\label{sec:exp_variance_models}

To measure the robustness of our approach, we consider an alternate
transmission model, and we measure whether, without making any changes,
our observer placement still performs well. For every edge $uv
\in E$ with weight $w_{uv}$, we take $X_{uv} \sim \mbox{Unif}([(1-\varepsilon)w_{uv},
(1+\varepsilon)w_{uv}])$. We find comparable results (see Appendix~\ref{app:figures}); they suggest that our observer 
placement is not dependant on the exact transmission model and that the
variance of the transmission delays is really a key factor for a good observer placement.

%

\section{Conclusion \& Future Work}
\label{sec:conclusion}

In this work, we have taken a principled approach towards budgeted observer placement for source localization. We are the first to have observed a dichotomy between the low and high-variance regimes, and we developed complementary approaches for both. 
We have evaluated our approaches against state-of-the-art and alternative heuristics and find that the performance of our algorithms is favourable.



One natural extension would account for two
stages of observation; in the first stage, as in this work, we select a set of
observers to monitor the network. In the next stage, once an epidemic begins,
we deploy additional observers in the relevant region of the
network. This would pave the way for other types of \emph{adaptive} models,
including ones where we not only \emph{observe} a node but can act to
\emph{immunize} it or in which we can \emph{move} the observers as required.\\


\section*{Acknowledgements}
B.Spinelli was partially supported by the Bill \& Melinda Gates Foundation, under Grant No. OPP1070273

\bibliographystyle{abbrv}
\small{
\bibliography{lnhnobs}
}

\clearpage

\newpage

\appendices
\section{Other Metrics}
\label{app:other_metrics}
In this section we define other metrics of interest. In particular, we consider
an adversarial setting (e.g., in the case of bio-warfare) where if our observers are known, the adversary would select the worst location for the source.

First, we consider minimum success probability, which is
\begin{eqnarray*}
 \hat \Ps(\mathcal{O}) &\defeq&  1- \max_{E_i} \left( \P(\hat{s} \neq s^* | s^*\in E_i) \right)  
\end{eqnarray*}
\noindent where $\{E_i\}$ are the equivalence classes with respect to $\mathcal{O}$. Note that in an adversarial setting, we would not consider any prior, rather would select $\hat s \in \argmax_{E_i} \P(\hat{s} \neq s^* | s^*\in E_i) $ uniformly at random;  given any non-uniform distribution, the adversary could place the source at the location with lowest probability. 


For the same reasons, we may also wish to consider the \emph{maximum distance} between the true and the estimated source as a metric.
\begin{eqnarray*}
\max (d(s^*, \hat{s})) &\defeq& \max_{s \in V} \Delta_s = \max_i \Delta_i, 
\end{eqnarray*}
\noindent where $\Delta_s$ (similarly $\Delta_i)$ denotes the diameter of equivalence class $[s]_\mathcal{O}$ (similarly $E_i)$. 
Note that, in particular, this is independent of any prior. 

Another natural consideration which interpolates between expected and worst-case metrics is the \emph{expected maximum distance} between the true and
the estimated source. This captures the case where there is a prior $Q$ on the
source, and we are able to identify the equivalence class of $s^*$, but make
the \emph{worst-case} estimation $\hat s$ within that class.
\begin{eqnarray*}
\E[\max(d(s^*, \hat{s}))] &\defeq& \sum_{s \in V} \P(s^* = s)  (\max_{u \in [s]_\mathcal{O}} d(s,u)) \\
&=& \sum_{s \in V} Q(s) \Delta_s =  \sum_{i} Q(E_i) \Delta_i .
\end{eqnarray*}

\section{Double Resolving Sets} \label{app:DRS_CHW}

The problem of \emph{minimizing} the required number of observers in order to perfectly identify the source in the zero-variance setting has been studied  \cite{ChenHW14}; 
an observer set $\mathcal{O}$ such that $\Ps(\mathcal{O}) = 1$ is called a Doubly
Resolving Set (DRS). While the original formulation of the DRS problem is slightly
different, this version follows straightforwardly from our observations in
Section~\ref{sec:low-variance}.

\begin{definition}[Double Resolving Set] 
Given a network $\G$, $S \subseteq
V$ is said to be a Double Resolving Set of $\G$ if for any $x, y \in V$ there exist
$u, v \in S$ s.t. $d(x, u) - d(x, v) \neq d(y, u)-d(y, v)$.
\end{definition}

\noindent Finding a Doubly Resolving Set of minimum size is known to be  
NP-hard~\cite{kratica}. An approximation algorithm, based on a greedy
minimization of an \emph{entropy} function, has been studied. 
Note that this has no connection to true information-theoretic entropy.

\begin{definition}[Entropy \cite{ChenHW14}]
Let $\mathcal{G}$ a network, $\mathcal{O} \subseteq V$, $|\mathcal{O}|=k$ a set of observers.
The entropy of $\mathcal{O}$ is 
\[H_\mathcal{O} = \log_2(\prod_{[u]_\mathcal{O} \subseteq V}|[u]_\mathcal{O}|! ).\]
\end{definition}

\noindent Note that $H_\mathcal{O}$ is minimized if and only if each equivalence class
consists of only one node and hence if and only if $\Ps=1$. 
%
%
However, despite the fact that $H_\mathcal{O}$ is minimized when $\Ps$ is maximized and that both act on the same set of equivalence classes for a given $\mathcal{O}$, the greedy processes that minimize $H_\mathcal{O}$ and maximize $\Ps$ are not 
the same. This can be seen by rewriting both objective functions in the following way. Let $[c_1, \ldots, c_q]$ be the sequence of equivalence class sizes.
Then $H_\mathcal{O}$ can be written as $H_\mathcal{O}([c_1, .., c_q]) = \sum_{i=1}^l \sum_{j=2}^{c_i} \log(j) =
\sum_{i=2}^{\max{c_j}} \log(i) \#\{c_j \geq i\}.$ 
Analogously we have the following equality for the success probability $\Ps([c_1, \ldots , c_q])$: $n (1-\Ps([c_1, \ldots, c_q])) = n - q = \sum_{i=2}^{\max{c_j}} \#\{c_j \geq i\}$
Hence, though similar in spirit, a greedy minimization of $H_\mathcal{O}$ is
not related to a greedy optimization of $\Ps$ (or $\E[d(s^*, \hat s)]$).

\section{Alternate objective functions}\label{app:alternate-obj}
\begin{table}[H]
\begin{center}
\begin{tabular}{l c c c}
\hline
\multicolumn{4}{c}{Random Geometric Graph, $N=100$, $r=0.2$}\\
\hline 
\\
& $\frac{\Ps(\Phi_{dist})-\Ps(\Phi)}{\Ps(\Phi)}$ &
$\frac{\E_d(\Phi_{dist})-\E_d(\Phi))}{\E_d(\Phi_{dist}) + 1}$ &
$\frac{\Ps(\Phi_{ent})-\Ps(\Phi)}{\Ps(\Phi)}$ \\
\\
\hline
$k = 2$ & -0.205 & -0.101 & -0.033\\

$k = 4$ & -0.014 & 0.003 & -0.007\\

$k = 8$ & -0.003 & 0.002 & -0.003\\


\hline
\\
\hline
\multicolumn{4}{c}{Barab\`asi Albert Graph, $N=100$, $m=3$}\\
\hline 
\\
& $\frac{\Ps(\Phi_{dist})-\Ps(\Phi)}{\Ps(\Phi)}$ &
$\frac{\E_d(\Phi_{dist})-\E_d(\Phi))}{\E_d(\Phi_{dist}) + 1}$ &
$\frac{\Ps(\Phi_{ent})-\Ps(\Phi)}{\Ps(\Phi)}$ \\
\\
\hline
$k = 2$ & -0.168 & -0.023 & -0.037\\

$k = 4$ & -0.039 & -0.025 & -0.028\\

$k = 8$ & -0.004 & 0.003 & 0.005\\


\end{tabular}\caption{Comparison of \textsc{lv-Obs} ($\Phi$) with the greedy algorithms that
minimize the entropy function of \cite{ChenHW14}
($\Phi_{ent}$) and the expected distance ($\Phi_{dist}$)}\label{table}
\end{center}
\end{table}

Here we compare Algorithm \ref{algo:budget}, denoted in this section as $\Phi$, with two
other greedy algorithms that allocate the budget for observers according to
different objective functions:
\begin{enumerate}
\item $\Phi_{ent}$ minimizes the entropy function $H_{\mathcal{O}}$ \cite{ChenHW14}
(see Section \ref{app:DRS_CHW});
\item $\Phi_{dist}$ minimizes the expected distance (see Equation
\eqref{eq:exp_dist}). 
\end{enumerate}
We considered different topologies and different budgets $k$ for the observers. The
results are given in the form of (averaged) relative differences in Table
\ref{table}. The standard error of measurement is not reported for the sake of
readability but it was checked to be small: approximately $10^{-2}$ for $k=2$
and $(\Ps(\Phi_{dist})-\Ps(\Phi))/\Ps(\Phi)$; on the order of
$10^{-3}$ or smaller in all the other cases.
Note that, since the expected distance can be $0$ we add $1$ in the
denominator when comparing $\E_d(\Phi_{dist})$ and $\E_d(\Phi)$.
The results achieved by these algorithms are, on average, worse than those of
Algorithm \ref{algo:budget} ($\Phi$) independently of the graph topology. The
only exception is the minimization of the expected distance when $k$ is very
small.

\section{Hardness of Budgeted Observer Placement}\label{app:hardness}

\newcommand{\A}{\mathcal A}

\begin{theorem}
Given a network $\G=(V,E)$ and a budget $k$, finding an observer set $\mathcal{O}$ which
maximizes $\Ps$ is NP-hard. 
\end{theorem}

\begin{proof}
We will prove that the budgeted observer placement is NP-hard with a reduction
from the DRS problem (see Section~\ref{app:DRS_CHW}), i.e., given a
polynomial-time algorithm for the budgeted
observer placement problem, we will prove that we can solve the DRS problem in
polynomial time. 

Assume that we have a polynomial-time algorithm $\A$ that takes as input a
network $\G = (V,E)$ and a budget $k$, and outputs a set $\mathcal{O} \subseteq V$ of size $k$
such that $\Ps$ is maximized. Recall from Section~\ref{sec:low-variance} that given a
network $\G$ and a set $\mathcal{O}$, the probability $\Ps$ can be calculated in time
$O(n)$ where $n = |V|$ (it is enough to compute the $n$ distances vector with
respect to $\mathcal{O}$ and any reference observer $o_1\in \mathcal{O}$). Hence,
we will construct an algorithm for the DRS problem. 

\begin{algorithm}[H]
\begin{algorithmic}
\caption{Finds the minimum cardinality $DRS$ given an algorithm
to compute the $k$-nodes set that maximizes $\Ps$.}
\Require{Network $\G=(V, E)$}
\For{$k = 1, \ldots, |V|$}
\State $\mathcal{O} := \A(G,k)$ 
\State $P := \Ps(\mathcal{O})$
\If{$P = 1$} 
\State \Return $k$
\EndIf
\EndFor
\end{algorithmic}
\end{algorithm}

Since the full set $V$ always resolves the network, the program is well defined
(i.e., it always returns \emph{some} $k$). Moreover, it returns precisely the
minimum budget $k$ required in order to attain $\Ps = 1$. Lastly, it is clear
that the runtime is at most $O(n(p_\A(n) + n))$ where $p_\A(n)$ is the running
time of algorithm $\A$. Hence, we have a polynomial-time algorithm for the DRS problem.
\end{proof}

\section{High-Variance Source Estimation}
\label{app:estimator}
Denote by $T_\mathcal{O}$ the observed infection process. If the infection
delays are Gaussian, $\G$ is a tree and no 
prior information about the source position is available, the
maximum likelihood (ML) estimator is defined as $\hat{s} \in \displaystyle\arg \max_{\substack{s\in V}} \P(s|T_\mathcal{O})$, which 
has a tractable closed form \cite{Pinto12}.\footnote{Note that the model of
\cite{Pinto12} additionally assumed infected observers knew the neighbor that infected them;  this assumption is not required for our work.} 
In particular, given a
set of observers $\mathcal{O}=\{o_1, o_2, \ldots, o_k\} \subseteq V$, the vector of observed infection delays
$\tau = [t_2-t_1, \ldots, t_k-t_1] \in \R^{k-1}$ is distributed as
$\mathcal{N}(d_{s, \mathcal{O}}, \mathbf{\Lambda}_\mathcal{O})$ 
where $d_{s, \mathcal{O}}$ is the distance vector of Definition~\ref{distance_vector} and the
covariance matrix $\mathbf{\Lambda}_\mathcal{O}$ is  
\begin{equation}\label{Lambda}
\boldsymbol\Lambda_{\mathcal{O}, (k,i)}=\sigma^2 \left\{\begin{matrix}
\sum_{(u,v) \in \mathcal{P}(o_1,o_{k+1})} w_{uv}^2 & k=i\\ 
\sum_{(u,v) \in \mathcal{P}(o_1,o_{k+1})\cap \mathcal{P}(o_1,o_{i+1})}  w_{uv}^2& k \neq i,
\end{matrix}\right.
\end{equation}
with $\mathcal{P}(x,y)$ denoting the set of edges in the unique path between node 
$x$ and node $y$. Hence the ML estimator is
\begin{equation}\label{ML_tree}
\begin{split}
\hat{s}&\in\displaystyle\arg \max_{\substack{s\in V}}\frac{\exp
\Big(-\frac{1}{2}(\tau-\mathbf{d}_{s,\mathcal{O}})^\top{\mathbf{\Lambda}_\mathcal{O}}^{-1}
(\tau-\mathbf{d}_{s,\mathcal{O}})\Big)}{|\mathbf{\Lambda}_\mathcal{O}|^{1/2}}\\
&=\displaystyle\arg \max_{\substack{s\in V}} 
\Big[\mathbf{d}_s^{\top}{\mathbf{\Lambda}_\mathcal{O}}^{-1}
(\tau-\frac{1}{2}\mathbf{d}_{s,\mathcal{O}})\Big].
\end{split}
\end{equation}

On non-tree networks, the multiplicity of paths linking any two nodes makes
source estimation more challenging. As claimed in \cite{Pinto12}, the same estimator can be used as an approximation of the ML estimator for a
non-tree network by assuming that the diffusion happens only through a BFS (\textit{Breadth-First-Search}) tree
rooted at the (unknown) source. In this case the paths which appear in the
definition of the covariance matrix $\mathbf{\Lambda}_{\mathcal{O}}$ are computed on the
BFS tree
rooted at the candidate source considered. Hence $\mathbf{\Lambda}_{\mathcal{O}}$ depends on the
candidate source  and the ML estimator is 

\begin{equation}\label{eq:estimator}
\hat{s}_{\mathrm{bfs}} \in\displaystyle\arg \max_{\substack{s\in V}}\frac{\exp
\Big(-\frac{1}{2}(\tau-\mathbf{d}_{s,\mathcal{O}})^\top{\mathbf{\Lambda}^s_\mathcal{O}}^{-1}
(\tau-\mathbf{d}_{s,\mathcal{O}})\Big)}{|\mathbf{\Lambda}^s_\mathcal{O}|^{1/2}}.
\end{equation}

In this work, we adopt \eqref{eq:estimator} as the source estimator in the noisy
case. In fact, even
if our edge delays are truncated Gaussians, under the
hypothesis of sparse observations, we can apply the Central Limit Theorem (CLT)
to approximate the sum of the edge delays with Gaussian random variables: if all
edges have the same weight we can apply the CLT for i.i.d. random variables; if
this is not the case, we can apply Lyapunov's version of CLT.\footnote{Lyapunov
condition with $\delta=1$ is easily verified for a sequence of independent and
uniformly bounded random variables (see Example $27.4$ in \cite{Billi} for more
details). 
}


%

\section{Additional Figures}\label{app:figures}

\begin{figure}[H]
\begin{center}
\subfigure[\textsc{lv-Obs}]{\includegraphics[width=.45\columnwidth]{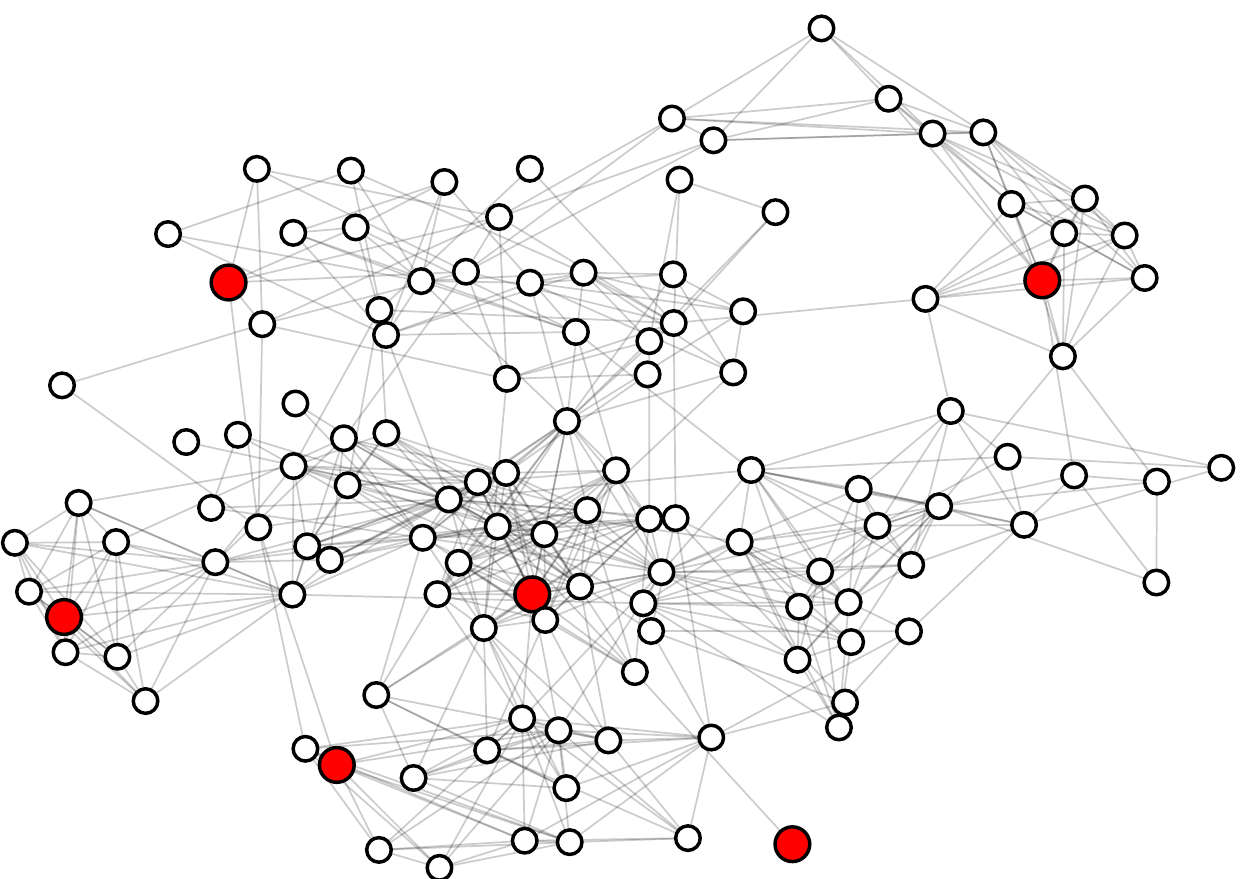}}
\subfigure[\textsc{hv-Obs},
$L/\tilde{\Delta}=0.5$]{\includegraphics[width=.45\columnwidth]{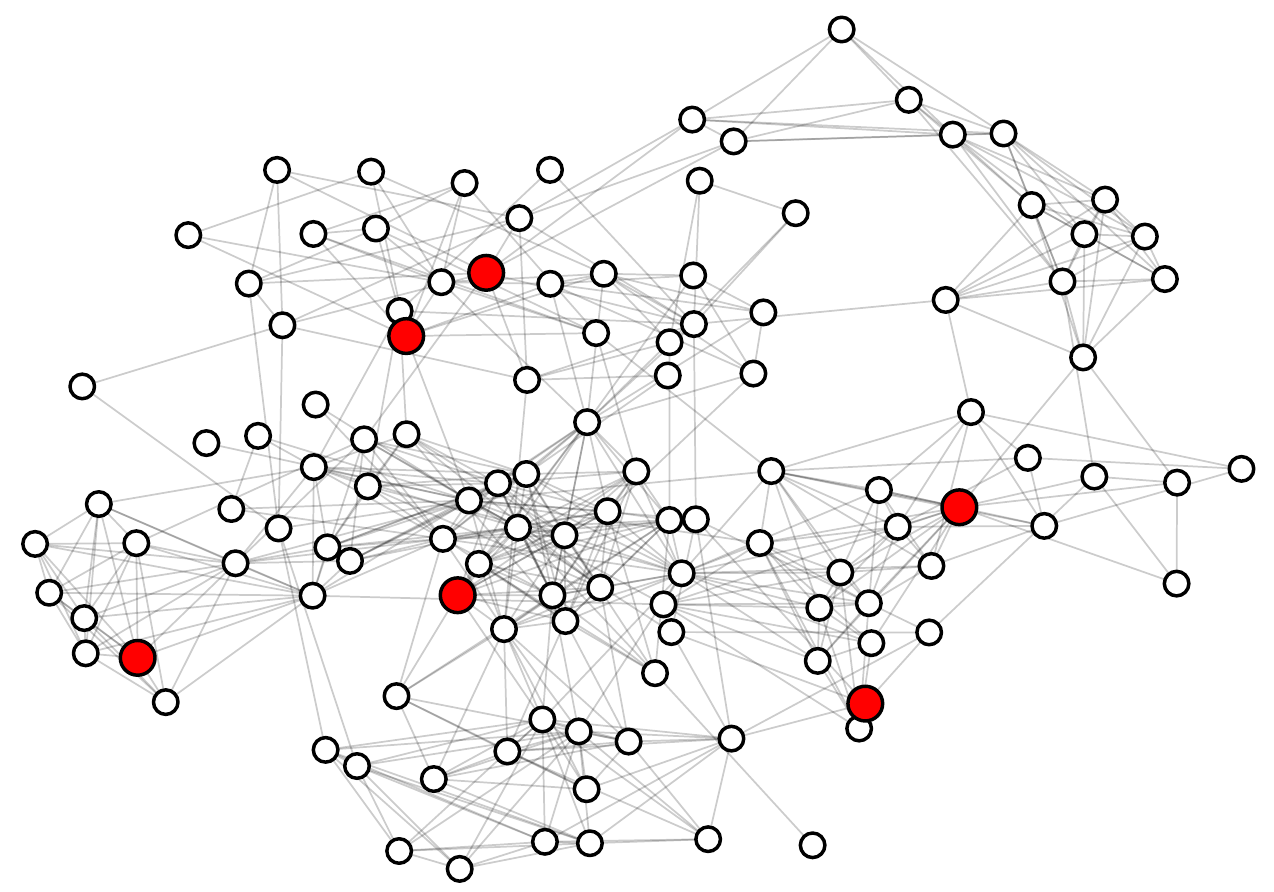}}
\caption{
The oberver placements of \textsc{lv-Obs} and
\textsc{hv-Obs} with $L/\tilde{\Delta}=0.5$ and $k=5\%$ on the F \& F network
are very different; \textsc{lv-Obs} contains leafs while
\textsc{hv-Obs} has shorter spacing.}\label{fig:ex_ff}
\end{center}
\end{figure}

\begin{figure}[H]
\subfigure[CR, 5\%
observers]{\includegraphics[width=0.47\columnwidth]{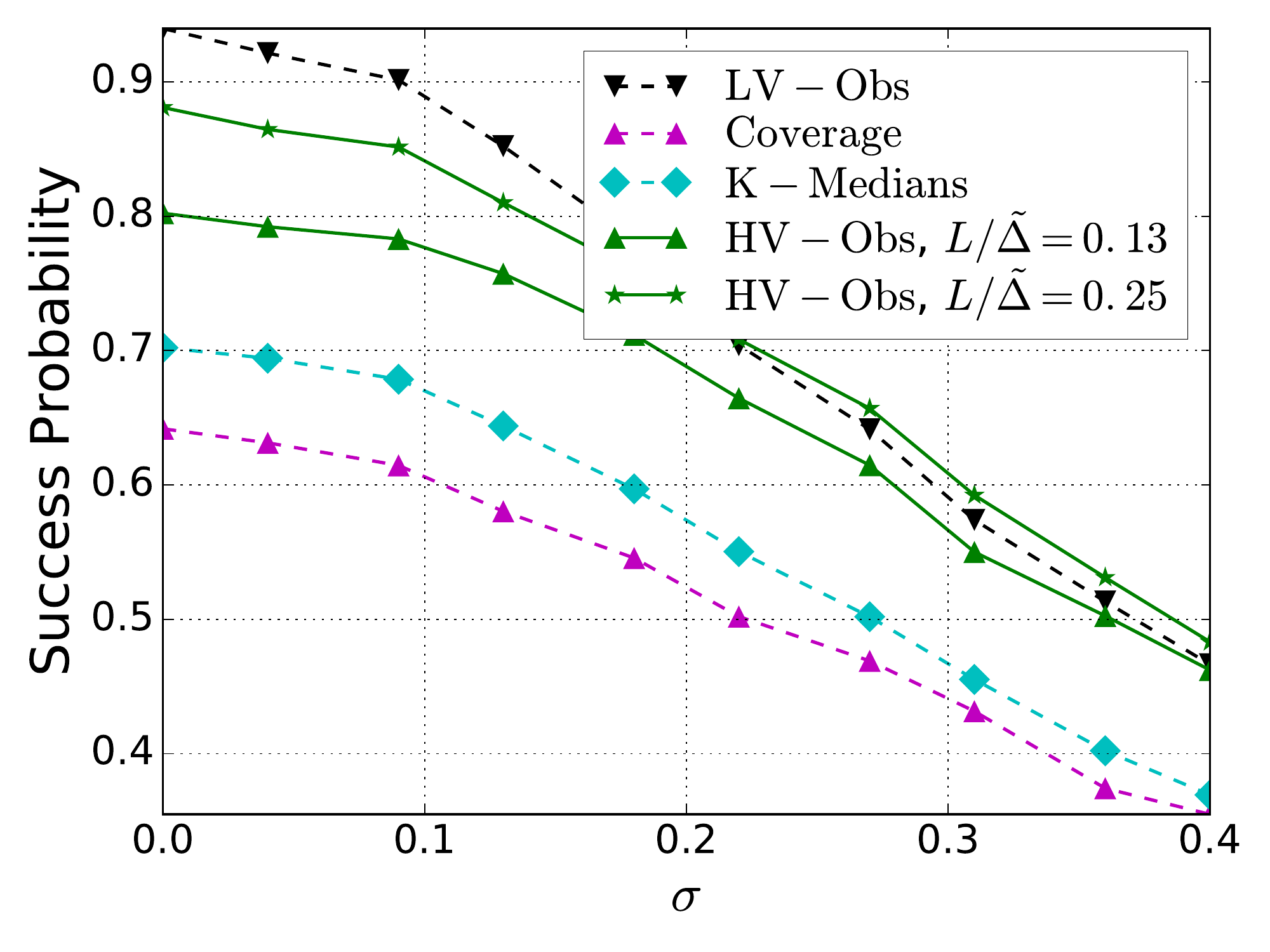}}
\subfigure[F \& F, 5\%
observers]{\includegraphics[width=0.47\columnwidth]{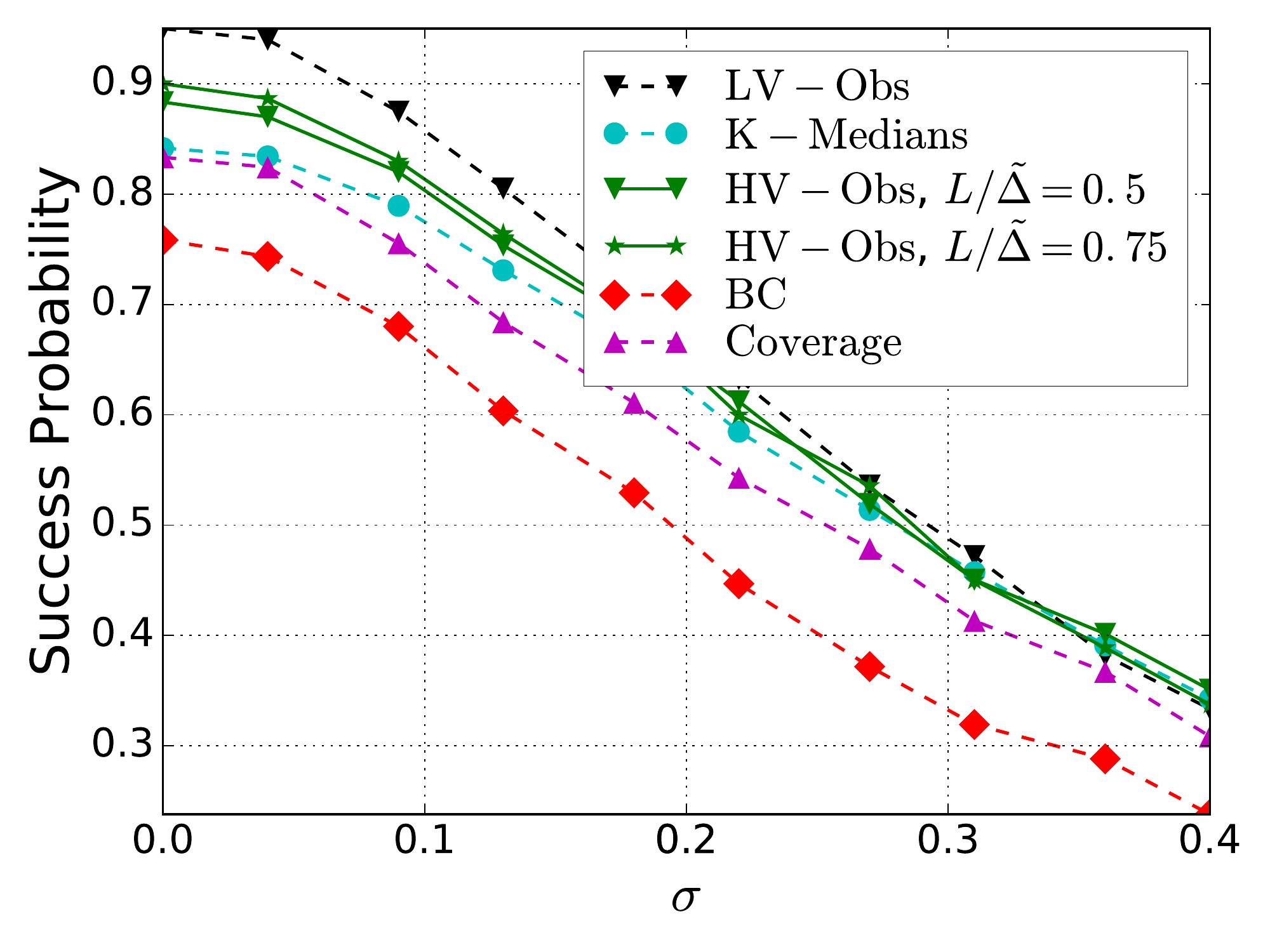}}\\
\subfigure[FB, 5\%
observers]{\includegraphics[width=0.47\columnwidth]{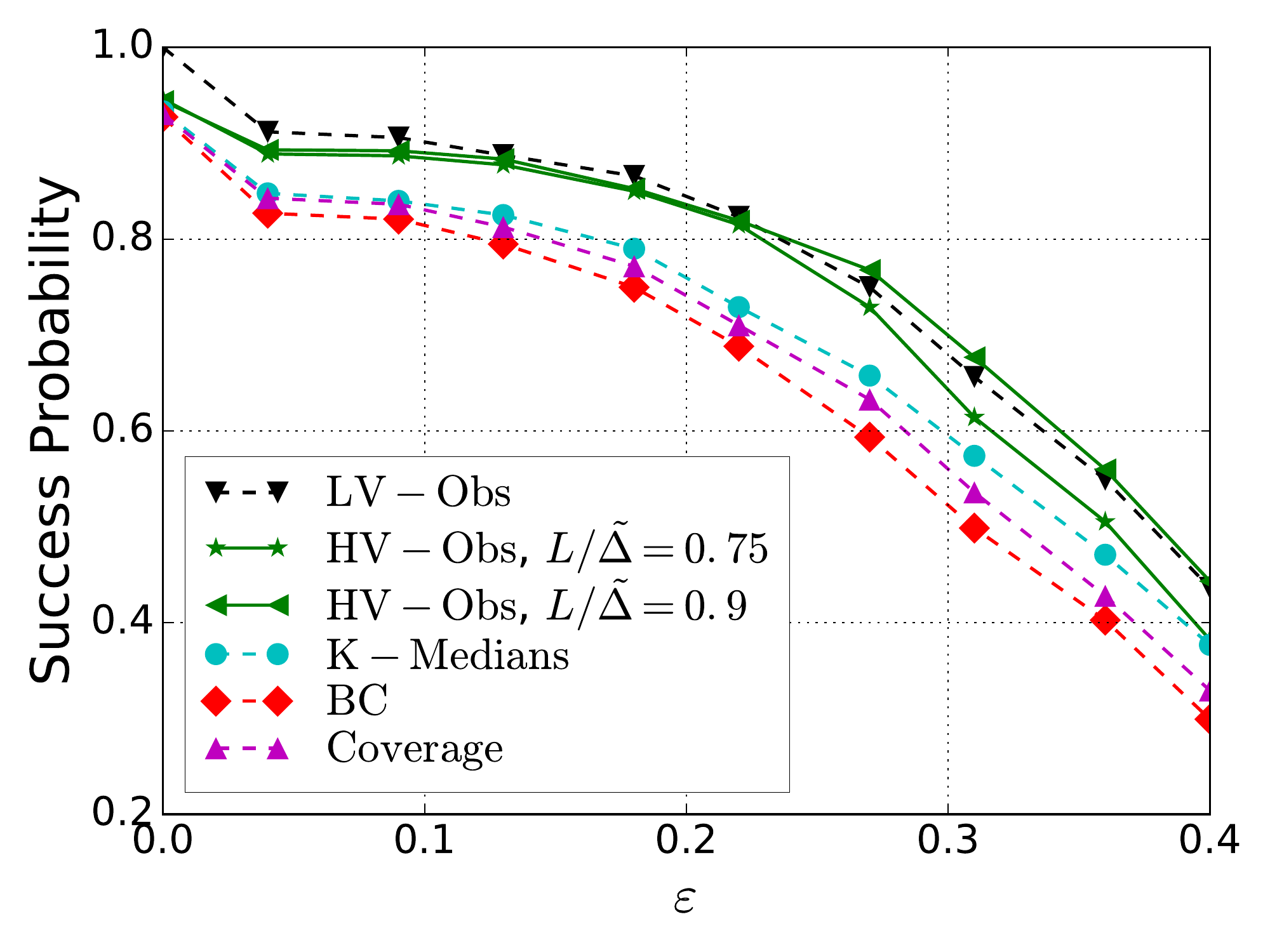}}
\caption{Success probability $\Ps$ as variance is increased on a uniform  transmission model (Section \ref{sec:exp_variance_models}).}
\label{fig:alt_model}
\end{figure}

\begin{figure}
\subfigure[California, 5\%
observers]{\includegraphics[width=0.47\columnwidth]{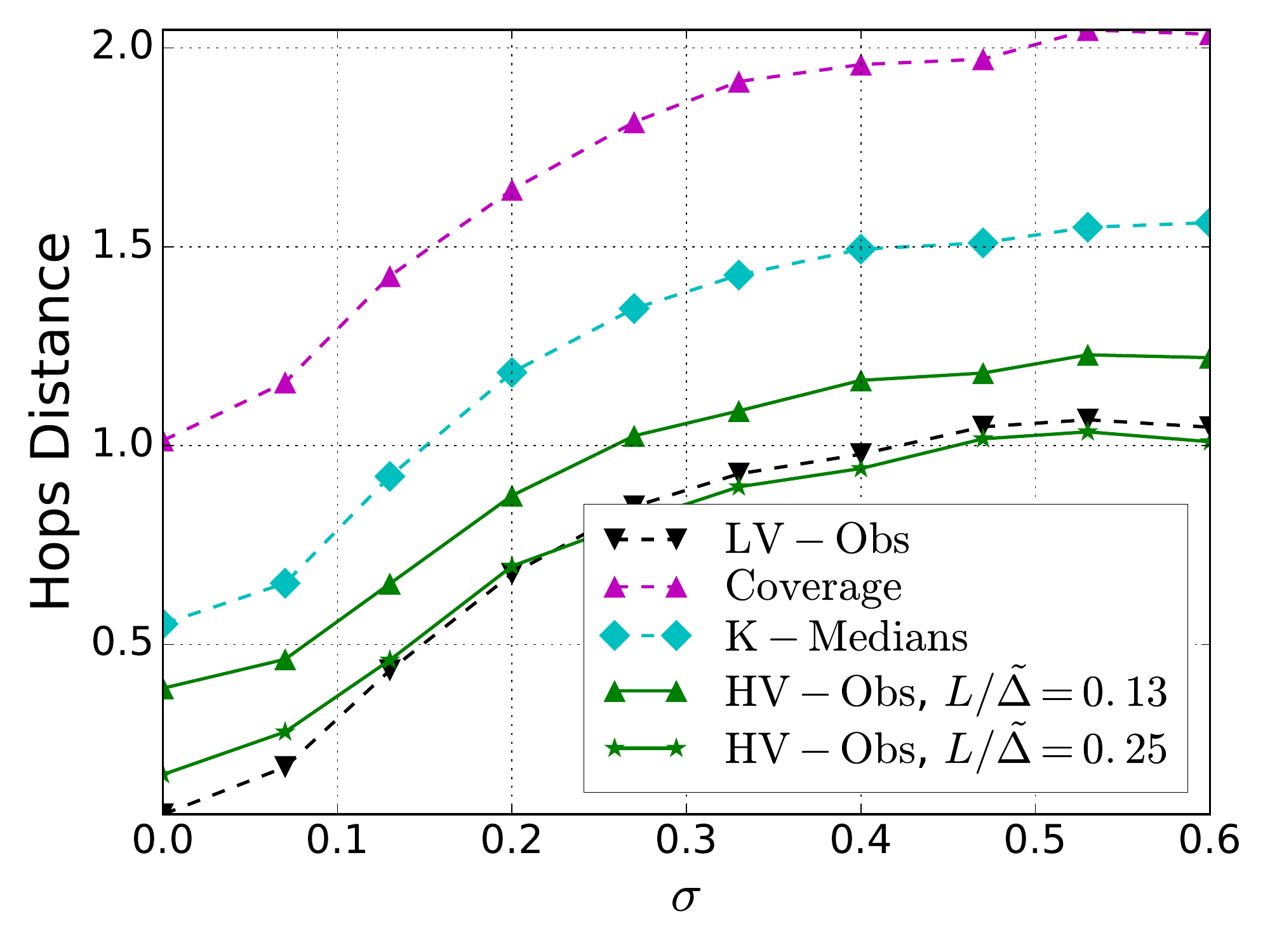}}
\subfigure[California, 5\%
observers]{\includegraphics[width=0.47\columnwidth]{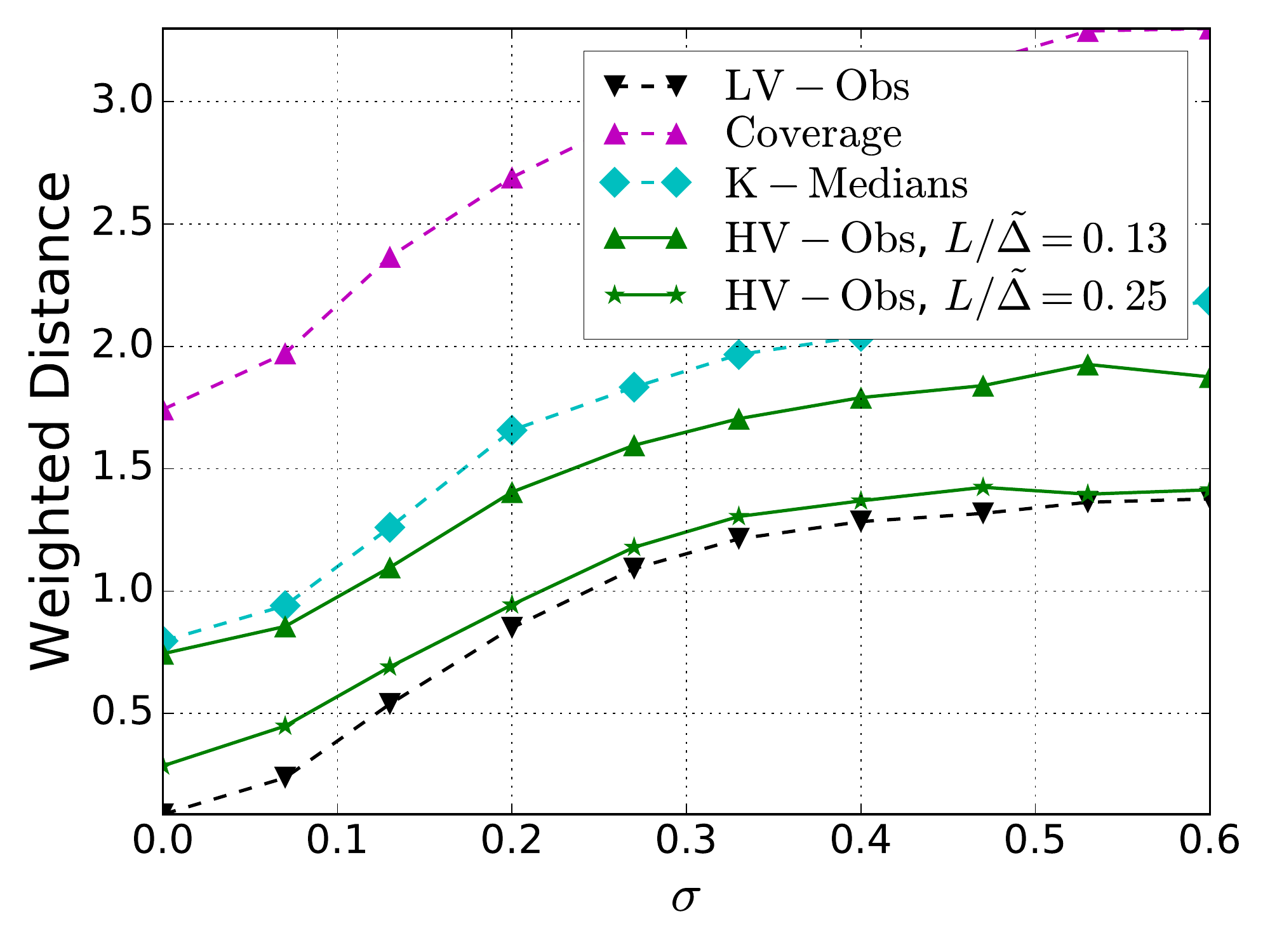}}\\
\subfigure[Friends \& Families, 5\%
observers]{\includegraphics[width=0.47\columnwidth]{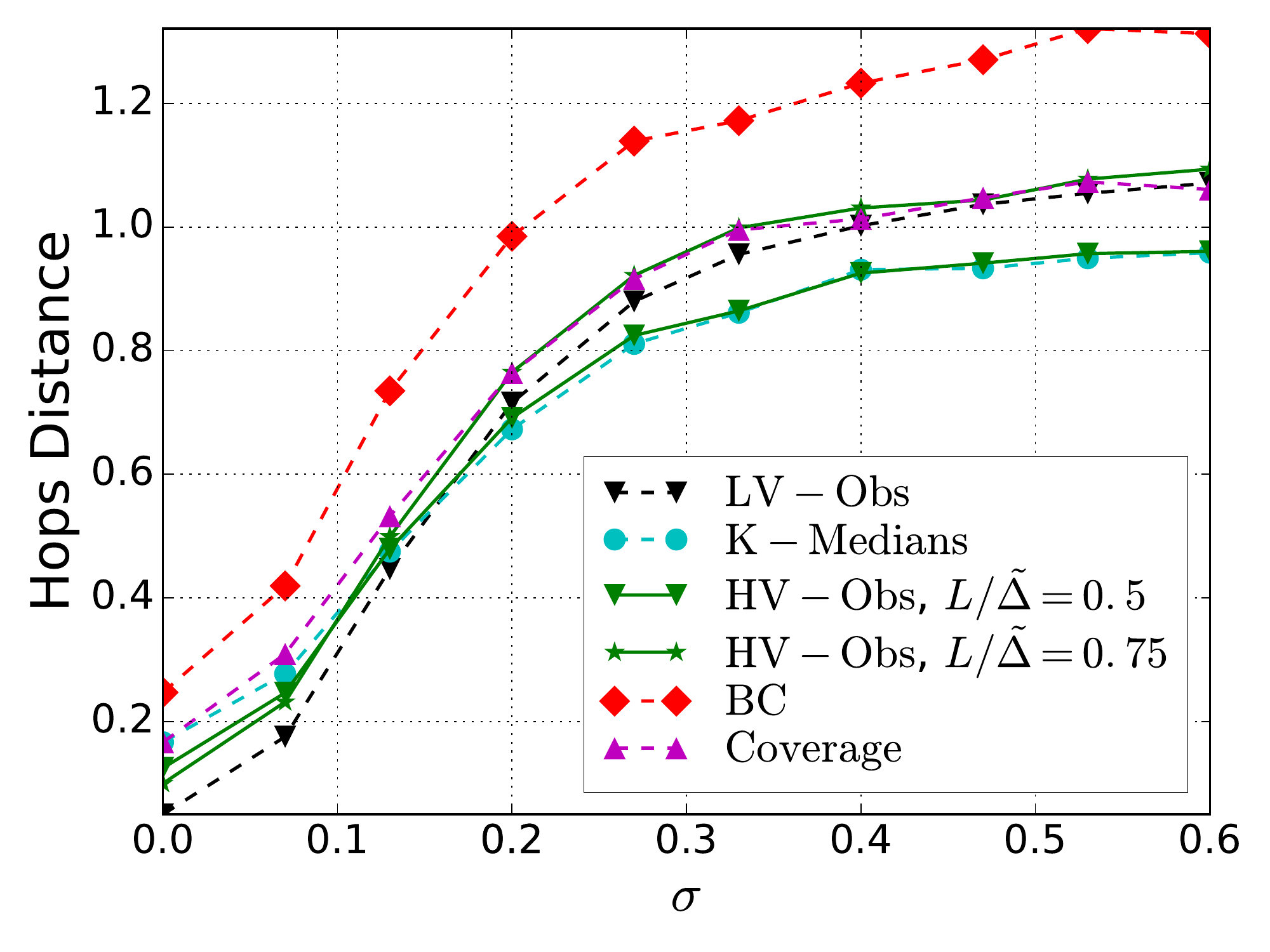}}
\subfigure[Friends \& Families, 5\%
observers]{\includegraphics[width=0.47\columnwidth]{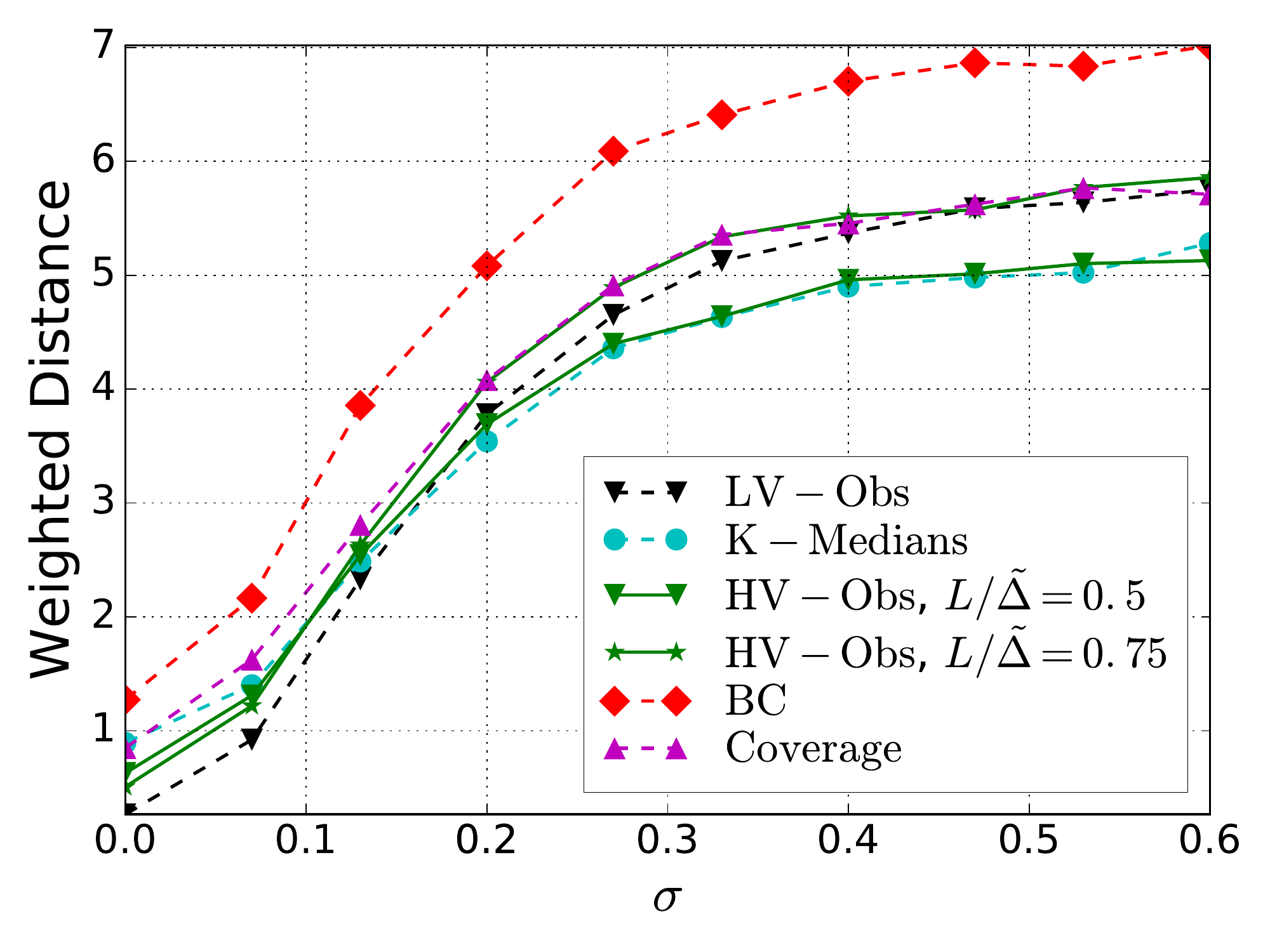}}\\
\subfigure[Facebook, 5\%
observers]{\includegraphics[width=0.47\columnwidth]{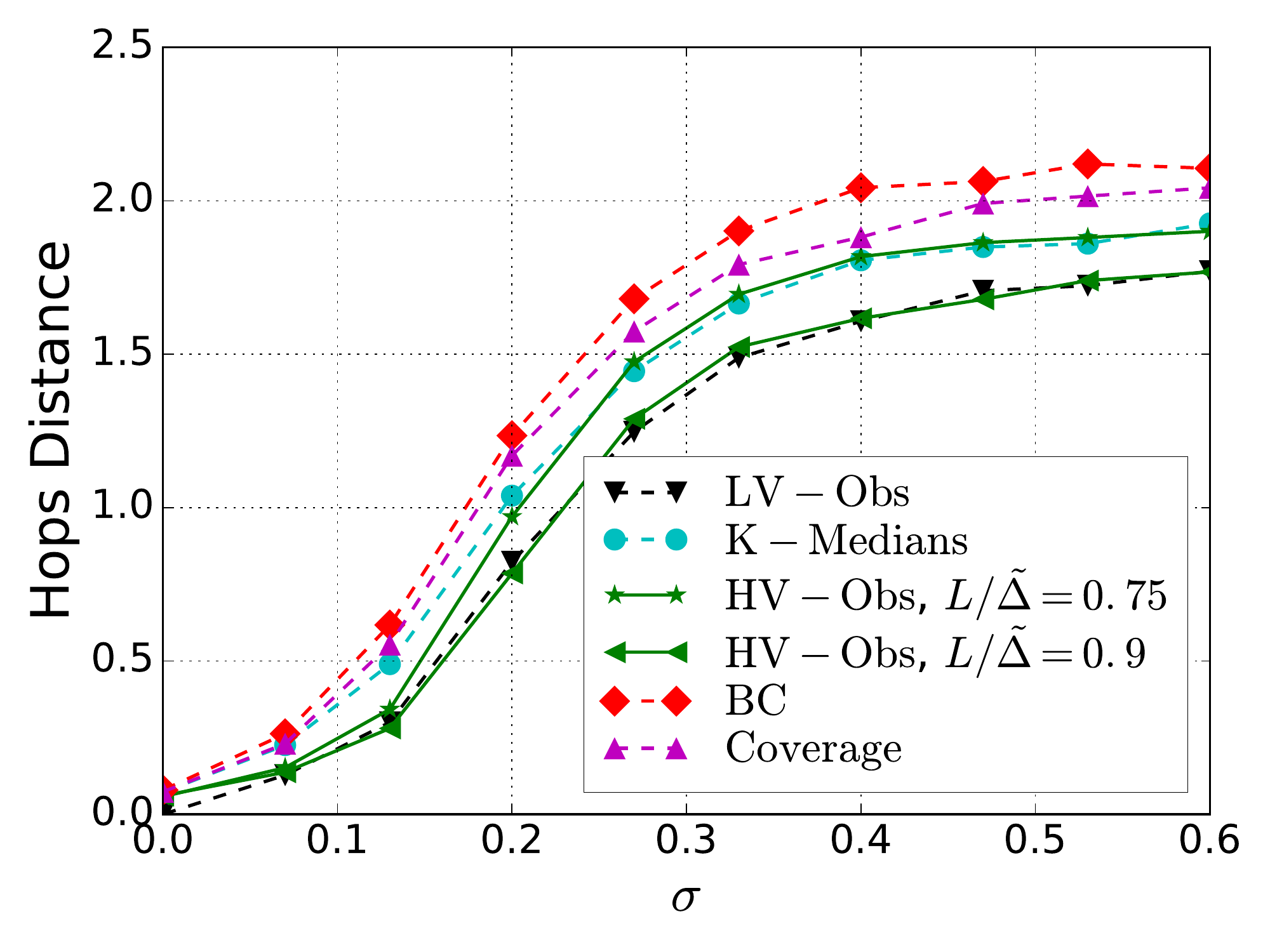}}
\subfigure[Facebook, 5\%
observers]{\includegraphics[width=0.47\columnwidth]{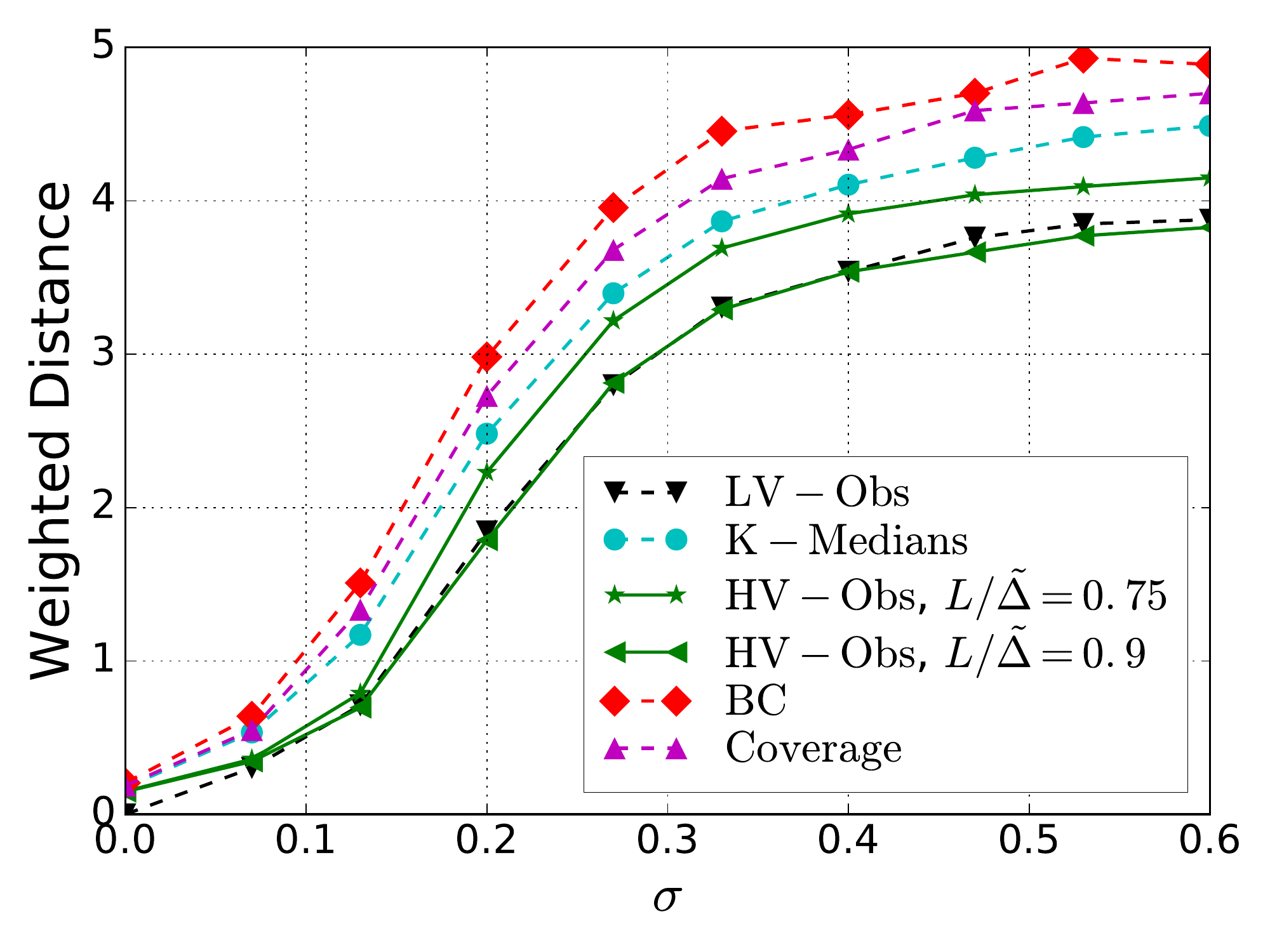}}
\caption{Expected distance in number
of edges (left column) and in weighted path length (right column) for the
datasets and setting of Section~\ref{sec:high-variance}}\label{fig:distance}
\end{figure}

\end{document}